\documentclass[%
 reprint,
nofootinbib,
 amsmath,amssymb,
 aps,
 onecolumn,
]{revtex4-2}
\usepackage[utf8]{inputenc} 
\usepackage[T1]{fontenc}    
\usepackage{hyperref}       
\usepackage{url}            
\usepackage{booktabs}       
\usepackage{amsfonts}       
\usepackage{nicefrac}       
\usepackage{microtype}      
\usepackage{xcolor}         
\usepackage{amsthm}

\usepackage{amsmath}
\usepackage{amssymb}
\usepackage{graphicx}
\usepackage{hyperref}
\usepackage{mathtools}
\usepackage{url}
\usepackage{footnote}
\usepackage{comment}
\usepackage{cprotect}

\usepackage{algorithm}

\usepackage{bbold}
\usepackage{algpseudocode}
\setcitestyle{numbers,square} 
\DeclareMathOperator{\E}{\mathbb{E}}
\DeclareMathOperator{\I}{\mathbb{I}}
\DeclareMathOperator{\R}{\mathbb{R}}
\DeclareMathOperator{\N}{\mathbb{N}}
\DeclareMathOperator{\PM}{\mathcal P}
\DeclareMathOperator{\argmax}{\arg\max}
\DeclareMathOperator{\argmin}{\arg\min}

\newcommand{\norm}[1]{\left\lVert#1\right\rVert}

\newtheorem{theorem}{Theorem}[section]

\newtheorem{lemma}[theorem]{Lemma}
\newtheorem{proposition}[theorem]{Proposition}
\newtheorem{definition}[theorem]{Definition}
\newtheorem{conjecture}[theorem]{Conjecture}

\usepackage{dcolumn}
\usepackage{bm}

\begin{document}

\preprint{APS/123-QED}

\title{Integer Traffic Assignment Problem: Algorithms and Insights on Random Graphs} 

\author{Rayan Harfouche}
\thanks{Equal contribution}

\author{Giovanni Piccioli}
\thanks{Equal contribution}

\email{giovannipiccioli@gmail.com}
\author{Lenka Zdeborová}
\affiliation{École Polytechnique Fédérale de Lausanne (EPFL)\\
  Statistical Physics of Computation Laboratory\\}

\begin{abstract}
Path optimization is a fundamental concern across various real-world scenarios, ranging from traffic congestion issues to efficient data routing over the internet. The Traffic Assignment Problem (TAP) is a classic continuous optimization problem in this field. This study considers the Integer Traffic Assignment Problem (ITAP), a discrete variant of TAP. ITAP involves determining optimal routes for commuters in a city represented by a graph, aiming to minimize congestion while adhering to integer flow constraints on paths. This restriction makes ITAP an NP-hard problem. While conventional TAP prioritizes repulsive interactions to minimize congestion, this work also explores the case of attractive interactions, related to minimizing the number of occupied edges.

We present and evaluate multiple algorithms to address ITAP, including the message passing algorithm of \cite{yeung2013physics}, a greedy approach, simulated annealing, and relaxation of ITAP to TAP. Inspired by studies of random ensembles in the large-size limit in statistical physics, comparisons between these algorithms are conducted on large sparse random regular graphs with a random set of origin-destination pairs. 
Our results indicate that while the simplest greedy algorithm performs competitively in the repulsive scenario, in the attractive case the message-passing-based algorithm and simulated annealing demonstrate superiority.  
We then investigate the relationship between TAP and ITAP in the repulsive case. We find that, as the number of paths increases, the solution of TAP converges toward that of ITAP, and we investigate the speed of this convergence. 

Depending on the number of paths, our analysis leads us to identify two scaling regimes: in one the average flow per edge is of order one, and in another the number of paths scales quadratically with the size of the graph, in which case the continuous relaxation solves the integer problem closely.
\end{abstract}
\maketitle
\section{Introduction}

\subsection{Motivation of the study}
Ranging from congestion problems in traffic networks to routing problems for data packages over the internet,
path optimization is ubiquitous in real-world problems. The traffic assignment problem (TAP) is one of the most prominent routing problems, with a long history dating back to the 50' \cite{doi:10.1680/ipeds.1952.11259}. 
In this work, we consider the integer traffic assignment problem (ITAP), a variant of TAP where the flows on edges are constrained to be integers. We are given a graph representing the roads in a city and a list of origin-destination (OD) pairs, each one representing a commuter that must move from one point to the other. For each OD pair we must find a path connecting the origin node to the destination node, while minimizing some congestion measure on the edges of the graph. Intuitively we want to distribute the paths (commuters) so that no edge (road) becomes too congested. While in TAP, for each OD pair, we are allowed to split the flow across several paths each carrying a fractional amount of the flow, in ITAP, each path must carry an integer amount of flow, that is, an integer number of commuters. This restriction makes ITAP an NP-hard problem. From a physics standpoint one can see this system as a system of variable-length interacting polymers, where each polymer is clamped at its endpoints and is either attracted to or repelled from others \cite{daoud1975solutions}. In TAP, the interaction between paths is commonly considered repulsive, the idea being that we want to minimize the congestion. In our work, we also consider the attractive case, where it is advantageous to merge paths as much as possible. This setting is motivated by network-building games in which one must build a graph that connects the OD pairs efficiently while minimizing the number of links used \cite{roughgarden2010algorithmic}. Think, for example, of a city that has to build road infrastructure that reaches every house efficiently while minimizing the length of constructed roads. In such formulation one clearly incentivizes paths to share edges, that is, avoiding building a road between every pair of houses. In another example, one aims to route internet traffic through as few nodes as possible so as to turn off the rest and save energy \cite{hinton2011power}.

Our study aims to provide a systematic comparison of various algorithms for both the repulsive and attractive versions of the problem. We perform this study on random instances of the problem where the underlying graph is sparse and random with a varying average degree, the origin-destination pairs are selected uniformly at random among all the nodes, and we vary the number of the paths considered. Further, we vary the size of the graph while keeping the average degree constant. Considering such random benchmarks is common in the field of statistical physics of disordered systems where physical systems \cite{mezard1987spin} are often modelled using a random realization of the disorder (here, the graph and the origin-destination pairs) and increasing the size of the system. 
Studying these random ensembles of instances then aims to identify properties that are prone to be universal in the sense that classes of real-world instances will exhibit analogous properties. The study of random instances of the ITAP problem was initiated in \cite{yeung2013physics}, and here we extend it by examining more algorithms.   

From an algorithmic point of view, we are particularly interested in the performance gap between simple greedy heuristics and more sophisticated algorithms, such as the message passing algorithm from \cite{yeung2013physics} or simulated annealing, that aim to take into account the interactions between the paths in a more detailed manner. Another very common manner to solve an integer optimization problem is to relax it to a version with continuous variables and investigate the performance gap between the relaxed and integer problems.   

\subsection{Problem definition}
In this section, we set the notation and define ITAP.
Let $G=(V, E)$ be an undirected graph. This is the graph over which paths will be routed.
Let $N=|V|$ be the number of nodes, which we index with numbers in $[N]\coloneqq\{1,2,\dots, N\}$. The origin-destination pairs 
$\{(s^\mu,t^\mu)\}_{\mu\in[M]}$ are also given. A path $\mu$ goes from its origin node $s^\mu$ to its destination node $t^{\mu}$. There are a total of $M$ paths, $\mu =1 \dots M$, routed over the network. 
We name $\pi^\mu$ the path joining $s^\mu$ and $t^\mu$.
Each path $\pi^\mu$  is represented as a sequence of adjacent nodes in $G$: $ \pi^\mu=(\pi^\mu (1)=s^\mu, \pi^\mu (2),\dots, \pi^\mu(L(\pi))=t^\mu)$, where $L( \pi^\mu)$ is the number of nodes in the path. Further, we use the expression $e\in \pi$ to indicate that $\pi$ traverses edge $e$.
The flow $I_e$ through an edge $e\in E$ is the number of paths that pass through $e$. Mathematically we can write this as $$I_e\coloneqq\sum_{\mu\in[M]} \mathbb{I}[e\in \pi^\mu]\, .$$
We aim to minimize the following energy function
\begin{equation}
\label{eq:energy}
    H( \pmb I)=\sum_{e\in E} \phi(I_e),
\end{equation}
We refer to $\phi:\R\mapsto[0,\infty)$ as the nonlinearity. $\phi(I_e)$ represents the total cost of having flow $I_e$ on edge $e$. We always assume $phi$ to be an increasing function. The cost is equally distributed among paths that pass through the edge: each one incurs in a cost $\phi(I_e)/I_e$. 

When $\phi$ is convex, minimizing $H$ will have a repulsive effect on the paths (this is the usual traffic scenario in TAP). If instead $\phi$ is concave, minimizing $H$ has an attractive effect on the paths, leading many paths to share the same edge. To see this, suppose $\phi(0)=0$, and
consider the cost of two paths (each carrying one unit of flow) sharing an edge vs. passing through different edges. In the first case the contribution to the energy is $\phi(2)$, while in the second case it is $2\phi(1)$. If $\phi$ is convex then $\phi(2)\geq 2\phi(1)$ therefore it is convenient to not share the edge. Vice versa if $\phi$ is concave the configuration with the paths overlapping is energetically convenient.

Several algorithms that we use to minimize the energy \eqref{eq:energy} rely on defining a Gibbs measure $\propto e^{-\beta H}$ over all possible paths and then studying its properties to find minima of $H$. More precisely, the measure is defined as 
\begin{equation}
    \label{eq:gibbs_measure}
    P_\beta\left(\{\pi^\mu\}_{\mu\in[M]}|\{s^\mu,t^\mu\}_{\mu\in[M]}\right)=\frac{e^{-\beta H(\pmb{I})}}{Z\left(\beta,\{s^\mu,t^\mu\}_{\mu\in[M]}\right)}\left[\prod_{\mu\in[M]}\I[\pi^\mu(1)=s^\mu]\I[\pi^\mu(L(\pi^\mu))=t^\mu]\I[\pi^\mu \text{ is a SAW}]\right].
\end{equation}
$\beta$ is an inverse temperature parameter: for $\beta=0$ the distribution is uniform across all self-avoiding walks (SAWs) with constrained endpoints, while for $\beta\to\infty$ it is supported on the global minima of $H$.
The three constraints ensure respectively that the path has the correct origin and destination and that it is self-avoiding, i.e., it does not visit the same node more than once. Finally $Z\left(\beta,\{s^\mu,t^\mu\}_{\mu\in[M]}\right)$ is a normalization constant.

The joint distribution of all the paths is very complex, and hard to work with in practice. Therefore many algorithms rely instead on considering one path (say $\pi^\nu$) at a time and approximating in different ways the conditional distribution of the path when all the others are fixed. We now derive the form of the conditional probability $P_\beta\left(\pi^\nu|\{\pi^\mu\}_{\mu\in[M]\backslash\nu},\{s^\mu,t^\mu\}_{\mu\in[M]}\right)$.

We first isolate the contribution of path $\nu$ in $H$. Define $I_e^{\backslash\nu}$ to be the flow on edge $e$ in absence of path $\nu$. Mathematically we have $I_e^{\backslash\nu}=\sum_{\mu=1, \mu\neq \nu}^M \I[e\in\pi^\mu]$. We can write $H$ in the following form.
\begin{align}
\label{eq:energy_decomp}
    H(\pmb I)&=\sum_{e\in E}\phi(I_e)=\sum_{e\in E} \phi(I^{\backslash\nu}+\mathbb{I}[e\in\pi^\nu])=\sum_{e\in E}\phi(I_e^{\backslash\nu})+\sum_{e\in E}\mathbb{I}[e\in\pi^\nu]\left(\phi(I_e^{\backslash\nu}+1)-\phi(I_e^{\backslash\nu})\right)=\\&=H(\pmb I^{\backslash\nu})+\Delta H(\pi^\nu; \pmb I^{\backslash\nu}),
\end{align}
with $\Delta H(\pi^\nu; \pmb I^{\backslash\nu})=\sum_{e\in E}\mathbb{I}[e\in\pi^\nu]\left(\phi(I_e^{\backslash\nu}+1)-\phi(I_e^{\backslash\nu})\right)$.
Notice that the only dependence of $\pi^\nu$ is contained in the term $\Delta H(\pi^\nu; \pmb I^{\backslash\nu})$. This implies that the conditional distribution has the form
\begin{align}
\label{eq:conditional_gibbs}
     P_\beta\left(\pi^\nu|\{\pi^\mu\}_{\mu\in[M]\backslash\nu},\{s^\mu,t^\mu\}_{\mu\in[M]}\right)=\frac{e^{-\beta \Delta H(\pi^\nu; \pmb I^{\backslash\nu})}}{Z^\nu\left(\beta,\pmb I^{\backslash\nu}, s^\nu,t^\nu\right)}\I[\pi^\nu(1)=s^\nu]\I[\pi^\nu(L(\pi^\nu))=t^\nu]\I[\pi^\nu \text{ is a SAW}],
\end{align}
where $Z^\nu\left(\beta,\pmb I^{\backslash\nu}, s^\nu,t^\nu\right)$ is again a normalization constant.
In words, $\Delta H(\pi^\nu; \pmb I^{\backslash\nu})$ is the energy landscape seen by path $\pi^\nu$ when all other paths are fixed. 

\subsection{Relaxation of the ITAP}
\label{sec:relaxed_problem_TAP}
In this section, we first explain what relaxation procedure transforms ITAP into TAP and then illustrate some properties of TAP that we will exploit in the numerical analysis. We start by defining the demand matrix $D\in\R^{N\times N}$, with elements $D_{xy}=\sum_{\mu\in[M]}\I[s^\mu=x]\I[t^\mu=y] $. In words, $D_{xy}$ counts the number of origin-destination (OD) pairs that are equal to $(x,y)$. For simplicity, we say that each path corresponds to one unit of flow. $D_{xy}$ then tells us how many units of flow must be routed from $x$ to $y$. An instance of TAP or ITAP is uniquely defined by giving the graph $G$ and the demand matrix $D$.
For each pair of nodes $(x,y)$ we define $\Pi_{xy}$ to be the set of paths going from $x$ to $y$. We indicate elements of $\Pi_{xy}$ with $\pi_{xy}^a$, $a\in\{1,\dots,|\Pi_{xy}|\}$. Similarly the flow carried by $\pi_{xy}^a$ is $h_{xy}^a$. We will refer to $\{h_{xy}^a\}$ as the path flows, and to $\{I_e\}_{e\in E}$ as the edge flows. We now give the definition of TAP and ITAP. 
\begin{equation}
\label{eq:TAP_ITAP_formal}
\begin{aligned}
& \underset{}{\text{minimize}}\qquad H(\pmb I) \\
&  \text{subject to}\\
&  I_e=\sum_{(x,y)\in(V\times V)} \sum_{a\in[|\Pi_{xy}|]} h^a_{xy} \I[e\in\pi^a_{xy}] \quad \forall e\in E\\
&  \sum_{a\in[|\Pi_{xy}|]} h_{xy}^a=D_{xy} \quad \forall (x,y)\in(V\times V)\\
&  h_{xy}^a \in\{0,1,\dots\} \quad \forall (x,y)\in(V\times V), \; \forall a\in[|\Pi_{xy}|] \textbf{ ITAP constraint}\\
&  h_{xy}^a \in[0,\infty) \quad \forall (x,y)\in(V\times V), \; \forall a\in[|\Pi_{xy}|] \textbf{ TAP constraint}
\end{aligned}
\end{equation}
The minimization is over all path flows $\{h_{xy}^a\}_{(x,y)\in V\times V, a\in [|\Pi_{xy}|]}$.

ITAP and TAP are identical except for the last constraint. In ITAP the flow on each path is constrained to be integer, while in TAP it can be any positive real number. Define $R_{xy}\coloneqq |\{a\in[|\Pi_{xy}|], \text{ s.t. } h_{xy}^a>0\}|$. This is the number of paths from $x$ to $y$ that carry a strictly positive flow. From the constraints, one sees that $R_{xy}\leq D_{xy}$ in ITAP, while it is potentially unbounded\footnote{that is bounded only by $|\Pi_{xy}|$ which generally scales super-exponentially in $N$} in TAP, as one can split the flow across all the paths.  In the following we will use the term 'at optimality' to mean that the paths and flow configuration corresponds to a global minimum of $H$.

\subsubsection{TAP is easy in the repulsive case}
Let us first consider the case in which $\phi$ is convex.
Below we list a few properties of TAP that make it a useful relaxation. See \cite{patriksson2015traffic,boyles2020transportation} for proofs.
\begin{theorem}
   Suppose $\phi:\R\mapsto\R$ is convex. Then:
\begin{enumerate}
    \item TAP is a convex optimization problem
    \item There exist polynomial (in $N,M$) time algorithms that solve TAP (see for example \cite{perederieieva2015framework})
    \item At optimality, the path flows $\{h_{xy}^a\}_{(x,y)\in V\times V, a\in [|\Pi_{xy}|]}$ are not in general unique
    \item At optimality, the edge flows $\{I_e\}_{e\in E}$ are unique    
\end{enumerate} 
\end{theorem}

The first two properties guarantee that TAP can be solved efficiently. The third property establishes that there can be several path flows that all give the minimal $H$. Last, the edge flows are instead unique when $H$ is minimal.

\subsubsection{TAP=ITAP in the attractive case}
\label{sec:TAP=ITAP_attractive}
In the previous section, we saw how ITAP relaxes to a convex optimization problem when $\phi$ is convex (repulsive case). This is no longer the case when $\phi$ is concave, as $H$ is not a convex function anymore and can, in fact, have multiple local minima. As a result, TAP is NP-hard for non-convex $\phi$.

Additionally, we now show that if $\phi$ is concave, then the solutions of TAP and ITAP coincide, meaning that even in the relaxed problem the optimal path flows are always integer.  First we define a local minimum in the problem TAP. We do so at the level of edge flows, since there can be several path flows corresponding to one edge flow.

We say $\pmb I$ is a TAP (ITAP) feasible flow, if it satisfies the three TAP (ITAP) constraints in \eqref{eq:TAP_ITAP_formal} for some choice of the path flows. Informally, a feasible edge flow, is one that is induced by some choice of path flows that respect the demand and positivity (and integrality) constraints.
\begin{definition}[TAP local minimum]
\label{def:tap_local_min}
    A TAP feasible edge flow $\pmb I$ is a local minimum of $H$ if for every feasible edge flow $\pmb{\tilde I}$, there exists an $\epsilon>0$ such that $H(\pmb I)\leq H(\epsilon\pmb{\tilde I}+(1-\epsilon)\pmb I)$.
\end{definition}
Notice that $\epsilon\pmb{\tilde I}+(1-\epsilon)\pmb I$ is feasible if both $\pmb{\tilde I},\pmb{I}$ are; so our definition says that at a minimum, moving in any feasible direction from $\pmb I$ results in an energy increase. We now give a characterization of local minima of $H$ when $\phi$ is concave. 
\begin{proposition}[]
\label{prop:no_flow_splitting_concave}
Consider an instance $G,D$ of TAP, with $\phi$ concave.  \textbf{Then} for every pair of nodes $x,y$  with $D_{xy}>0$, at any local minimum of $H$ all the flow going from $x$ to $y$ must be routed through a single path. Put differently $R_{xy}=1$.
\end{proposition}
\begin{proof}
   Let $\pmb I^{\backslash xy}$ be a feasible flow for the problem with $D_{xy}=0$ and the rest of the demand matrix unchanged. Put differently, $\pmb I^{\backslash xy}$ is a flow that results from routing all the flow except for that going from $x$ to $y$. Let $\pi^0, \pi^1$ be two distinct paths joining $x$ and $y$. Consider the following two flows
   \begin{equation}
       I^0_e=I^{\backslash xy}_e+D_{xy}\mathbb I[e\in \pi^0], \quad I^1_e=I^{\backslash xy}_e+D_{xy}\mathbb I[e\in \pi^1],\qquad e\in E
   \end{equation}
   Both $\pmb I^0$ and $\pmb I^1$ are feasible flows for $(G,D)$. All the flows $\pmb I^\lambda$ in the interpolating sequence $\pmb I^\lambda=\lambda \pmb I^1+(1-\lambda)\pmb I^0, \; \lambda\in[0,1]$ are also feasible. In $\pmb I^\lambda$, $\lambda D_{xy}$ units of flow are sent through $\pi^1$ and the rest is sent through $\pi^0$, thus the flow is split among two paths.
   Using concavity we have that $H(\pmb I^\lambda)$ is concave as a function of $\lambda\in[0,1]$. This implies that it can have local minima only for $\lambda\in \{0,1\}$ corresponding to all of the flow being routed through a single path.
\end{proof}
Proposition \ref{prop:no_flow_splitting_concave} immediately implies that also on global optima all the flow between any two nodes gets routed through a single path. As a consequence, since in our case the matrix $D$ has integer entries, a solution to TAP also satisfies the ITAP constraint in \eqref{eq:TAP_ITAP_formal} and is therefore integer.

\subsection{Related Literature}
When the objective function is convex, ITAP is related to TAP, a problem consisting of routing paths (users) over a network in a way that minimizes the average time taken to reach the destination. The study of TAP dates back to the seminal work of Wardrop \cite{doi:10.1680/ipeds.1952.11259} who formulated a mathematical model of traffic assignment and introduced the concepts of user equilibrium and system optimum. These correspond respectively to the case where each path egoistically chooses the route that minimizes \emph{its own} travel time, and the case where the route taken by each path is mandated so to minimize the total travel time of all paths. Wardrop's principles state the mathematical conditions for user equilibrium to be satisfied and establish that both problems can be formulated as optimization problems but with different objectives. Subsequently, \cite{beckmann1956studies,dafermos1969traffic} studied TAP as a convex optimization problem. The first algorithm employed to solve it is the Frank-Wolfe method \cite{frank1956algorithm}, followed by increasingly efficient algorithms \cite{bar2002origin, dial2006path, wei2019efficient, weintraub1985accelerating, mitradjieva2013stiff}. See \cite{patriksson2015traffic,boyles2020transportation} for a comprehensive review of TAP, and \cite{perederieieva2015framework} for a comparison of all of the common optimization algorithms to solve it. 

The integer version (ITAP) of the problem also received some interest, for example in \cite{erlander1988relationship} Wardrop's principles and the optimization formulation were extended to ITAP. 
ITAP is also related to the multicommodity flow problem (MCF) and the integer multicommodity flow problem (IMCF), in which, analogously to ITAP, one cannot split the flow among several paths \cite{barnhart2000using,even1975complexity}.
The difference with TAP is that in MCF and IMCF the interaction between paths is given exclusively by a capacity constraint, i.e., the number of paths crossing an edge cannot exceed its capacity \cite{salimifard2022multicommodity,shahrokhi1990maximum}. Variants of ITAP, such as the routing and wavelength assignment problem, arise in optics and telecommunications, where one must establish connections between pairs of nodes in a fibre optics network \cite{ozdaglar2004optimal,meli2004simulation,chlamtac1992lightpath}.

TAP traditionally only considers the case in which the paths repel each other. From a traffic perspective this makes sense, since we do not want too many users on the same road. The ITAP problem we consider also allows for attractive interactions between paths. This line of work has been studied in the setting of network building games. For a review of congestion and network building games see \cite{roughgarden2010algorithmic}. The game theory analysis is more concerned with quantifying the price of anarchy in games, that is the inefficiency of user (Nash) equilibria in the game where each user acts egoistically, with respect to the system optimum \cite{roughgarden2005selfish, pigou1920economics, koutsoupias1999worst, anshelevich2008price}. While finding Nash equilibria is a seemingly different problem, Rosenthal proved that one can build a potential function with the property that every minimum of the potential corresponds to a Nash equilibrium \cite{rosenthal1973class}. Therefore finding a Nash equilibrium consists of finding a local minimum of the ITAP objective function, for an appropriate choice on the nonlinearity. The game theory approach can also be useful in designing algorithms for ITAP, for example the greedy algorithm is inspired to the best response dynamics in congestion games \cite{fabrikant2004complexity,fanelli2008speed}. Also the logit dynamics for congestion games, which has been studied in \cite{blume1993statistical, kleer2021sampling} has some similarity with the simulated annealing algorithm we propose.

The statistical physics approach to ITAP has been initiated in \cite{yeung2013physics}. In this work the authors study ITAP on sparse random regular graphs (RRGs). They identify one of the relevant parameters of the problem (i.e. they find the $N,M$ dependence of $\rho$ in \eqref{eq:rho}) and propose the conditional belief propagation (CBP) algorithm (our name). CBP is a belief propagation \cite{pearl1982, peierls1936statistical,thouless1977solution} based algorithm derived using the replica method \cite{mezard1987spin,castellani2005spin}. See \cite{mezard2009information, zdeborova2016statistical, bishop2006pattern} for a comprehensive treatment of message passing and replica techniques. In our analysis, we will borrow the RRG setting and provide a shorter and more transparent derivation of CBP. 
The simulated annealing \cite{kirkpatrick1983optimization} algorithm we use updates the paths by proposing at each time step a randomly generated SAW with fixed endpoints. Several  Markov chain Monte Carlo (MCMC) methods for sampling unconstrained  SAWs have been developed \cite{van2015statistical, madras1988pivot,tabrizi2015rosenbluth}. Knuth \cite{knuth1976mathematics} devised an algorithm to sample SAWs with fixed endpoints on a grid, but it samples from a distribution that is nonuniform and difficult to characterize.
We devise a variant of Knuth's algorithm that allows us to sample SAWs from a range of distributions on a graph. In particular, by tuning a temperature parameter, we can go from sampling small fluctuations around the shortest path, to sampling SAWs uniformly.

To conclude, the integrality of the solutions to the relaxed problem, has been studied in a game theoretical setting in \cite{kleer2021computation}. However the integrality results hold under quite restrictive conditions, which are not satisfied in our case.

\subsection{Our contribution}
Below we list our main contributions.
\begin{itemize}
    \item We propose two novel algorithms to solve ITAP. The first, RITAP, is a simple algorithm based on relaxing ITAP to TAP, solving TAP, and then projecting the solution back on the integer constraints. Its strength lies in the computational efficiency when $M$ is large and in the possibility to study it theoretically. 
    The second algorithm is based on simulated annealing, a standard technique to solve non-convex optimization problems. The advantage of this method originates from the design of the underlying MCMC, which is based on a novel algorithm that samples SAWs with fixed endpoints. Numerical results show that simulated annealing is often the most effective method, especially in the attractive case.
    Finally, the SAW sampler is also of independent interest to the community studying techniques to generate SAWs randomly. 
    \item We present a greatly simplified (with respect to \cite{yeung2013physics}) derivation of CBP that emphasizes the relation with the greedy algorithm.
    \item We compare the performance of the considered algorithms and study the properties of ITAP in a random instance setting. The random instances of ITAP consist of sparse RRGs, uniformly random OD pairs and nonlinearities of the form $\phi(x)=x^\gamma$. The comparison is carried out in the regime where the flow per edge is of order one. We find that in the repulsive case, most algorithms including the simple greedy one, exhibit comparable performance. In the attractive case the CBP and simulated annealing algorithms have an edge over the others. 
    \item We study the relation between TAP and ITAP, as the number of paths grows.
    In the random instance setting, we observe that for convex $\phi$, when $M\gg N^2$, TAP and ITAP reach identical energies, moreover TAP solutions almost satisfy the integrality constraints. This implies that ITAP is easy to solve in this limit. We conjecture that this phenomenology holds for convex nonlinearities, provided that the distribution of OD pairs is sufficiently uniform. A formal statement of this conjecture is presented in \ref{conj:limit_degen}.
    \item In random instances, as we change the number of paths $M$ on the graph, we identify and study two intrinsic regimes of the problem, defined by the parameters $\rho$ and $\eta$ (see respectively \eqref{eq:rho},\eqref{eq:eta}). The first, partly identified in \cite{yeung2013physics}, arises when $M$ is such that the average flow per edge is of order one. In this regime, the effective interaction between paths arises and is usually the strongest, in the sense that the optimized energy found by the algorithms differs the most from the energy of the shortest paths, which are the optimal paths in the absence of interactions.
    The second regime is encountered when $M=\Theta(N^2)$, and thus the flow per edge is very large. When $M\approx N^2$ we observe the onset of an asymptotic behavior that continues for higher $M$. It is precisely in this scaling that the TAP solution (in the convex case) approaches the ITAP one. The presence of the two regimes is verified also in the case of graphs topologies taken from real-world traffic networks.
    
\end{itemize}
\section{Algorithms}
Below we state the algorithms we studied to solve ITAP. We release the code of each algorithm at \url{github.com/SPOC-group/algorithms_integer_traffic_assignment} \footnote{To request the code of the CBP algorithm as implemented in \cite{yeung2013physics}, please write to Bill Chi Ho Yeung.} and the code and data to reproduce our results at \url{github.com/SPOC-group/numerics_integer_traffic_assignment}.
\subsection{Greedy Algorithm}
\label{sec:greedy_algo}
The greedy algorithm is perhaps the simplest algorithm that can be devised to minimize the energy \eqref{eq:energy}. First the paths are arbitrarily initialized. At every time step, a path $\nu\in[M]$ is chosen. Then the path is updated to the value that minimizes $H$, leaving all other paths fixed. Since we want to minimize $H$ with respect to $\pi^\nu$, by the decomposition \eqref{eq:energy_decomp} this is equivalent to minimizing $\Delta H(\pi^\nu; \pmb I^{\backslash\nu})$. In turn this coordinate minimization is equivalent to finding the shortest path from $s^\nu$ to $t^\nu$ in a graph where the weight of edge $e$ is $\left(\phi(I_e^{\backslash\nu}+1)-\phi(I_e^{\backslash\nu})\right)$. This is easily doable for example via Dijkstra algorithm \cite{dijkstra2022note,cormen2022introduction}. The greedy algorithm proceeds by iteratively selecting an index $\nu$, then setting $\pi^\nu$ to be the shortest path in the aforementioned graph, until $H(\pmb I)$ stops decreasing. In appendix \ref{app:pseudocode_greedy} we provide the pseudocode of an efficient implementation of this algorithm.
One can see the greedy algorithm as a method that samples from the conditional distribution \eqref{eq:conditional_gibbs} after setting $\beta=\infty$. In the infinite $\beta$ limit the conditional distribution becomes in fact concentrated on the shortest paths.

\subsection{Conditional Belief Propagation (CBP)}
\label{sec:CBP}
Belief propagation (BP) is a message passing algorithm that allows to approximate the marginals of a given probability measure. It does so by exchanging messages between the nodes of the factor graph associated to the measure.
One can attempt to apply BP to the joint Gibbs measure of all paths \eqref{eq:gibbs_measure}. In this case BP would allow to approximate the probability that a path passes through a given node, from which the paths can then be reconstructed.
Already in \cite{yeung2013physics}, it was recognized that this approach is unfeasible: the fact that potentially all the paths interact on each edge results in an exponential time complexity in $M$. 
To circumvent this problem the CBP algorithm was proposed in \cite{yeung2013physics}. The idea is to apply BP to the conditional distribution \eqref{eq:conditional_gibbs} of a single path, while fixing all the others, and then iterating through paths. In this section we sketch CBP algorithm; Appendix \ref{app:BP_derivation} contains the details of its derivation.
We indicate with $i\to j$ the directed edge from node $i$ to $j$. In the graph $G$, every undirected edge is treated as a pair of opposite directed edges. Path $\pi^\nu$ is  treated as a directed sequence of edges starting at $s^\mu$ and terminating at $t^\nu$. We first introduce the binary variables $\sigma_{i\to j}^\nu\coloneqq \mathbb I[i\to j\in \pi^{\nu}]$, which indicate the directed edges through which $\pi^\nu$ passes. For an undirected edge $ij$ we define $H_{ij}^{\nu}\coloneqq \phi(I_{ij}^{\backslash\nu}+1)-\phi(I_{ij}^{\backslash\nu})$\footnote{ Recall that we do not separate flow on an edge by direction. $I_e$ is the total bidirectional flow on the undirected edge $e$}. This way we can write the conditional \eqref{eq:conditional_gibbs} as
\begin{align}
\label{eq:conditional_prob_with_edgecosts}
    P\left(\{\sigma^{\nu}_{i\to j}\}_{(ij)\in E}|\{H^{\nu}_{ij}\}_{ (ij)\in E},s^{\nu},t^{\nu}\right)\propto\left[\prod_{i\in [N]\backslash s^\nu,t^\nu} \mathbb{I}\left[\sum_{\substack{k,l\in \partial i\\k\neq l}}\sigma^{\nu}_{k\to i}+\sigma^{\nu}_{i\to l}=(0 \text{ OR } 2)\right] 
    \exp{\bigg(-\beta \sum_{k\in \partial i} H^{\nu}_{ik} \sigma^{\nu}_{i\rightarrow k} \bigg)} \right] \times \nonumber \\ \times \mathbb{I}\left[\sum_{k\in\partial s^\nu}\sigma^{\nu}_{s^\nu\to k}=1\right]\mathbb{I}\left[\sum_{k\in\partial s^\nu}\sigma^{\nu}_{k\to s^\nu}=0\right]\mathbb{I}\left[\sum_{k\in\partial t^\nu}\sigma^{\nu}_{k\to     t^\nu}=1\right]\mathbb{I}\left[\sum_{k\in\partial t^\nu}\sigma^{\nu}_{t^\nu\to k}=0\right] 
    \prod_{k \in \partial s^\nu}  \bigg( \exp{(-\beta H^{\nu}_{s^\nu k} \sigma^{\nu}_{s^\nu \rightarrow k})} \bigg),
\end{align}
where $\partial i$ denotes the set of neighbours of node $i$. The constraint in the first line ensures that the path passes at most once in each node. 
The constraints in the second line take care of the origin and destination nodes, which represent special cases. 
From this probability distribution, BP equations are derived following the procedure described in \cite{mezard2009information,zdeborova2016statistical}. We take the limit $\beta \rightarrow \infty$ and obtain the following update equations for two sets of messages 
$a^{\nu}_{i\rightarrow j}$ and $b^{\nu}_{i\rightarrow j}$. 
\begin{align}
\label{eq:BP_a}
a^{\nu}_{i\to j}&=\left\{\begin{array}{lr}
         \min_{k\in\partial i\backslash j} [a^{\nu}_{k\rightarrow i}]  + H^{\nu}_{ij} - \min \bigg[0, \min_{\substack{k,l\in\partial i\backslash j\\k\neq l}}[a^{\nu}_{k\rightarrow i}+b^{\nu}_{l \rightarrow i}] \bigg], & \text{if } i\not \in \{s^{\nu},e^{\nu}\}\\ 
         \infty & \text{if } i=e^{\nu}  \\ \\
         H^{\nu}_{s^\nu j}-\min_{k\in \partial s^\nu\backslash j } [b^{\nu}_{k\rightarrow s^\nu}], & \text{if } i=s^{\nu}  
        \end{array}\right. \\
\label{eq:BP_b}
b^{\nu}_{i\to j}&=\left\{\begin{array}{lr}
         \min_{k\in\partial i\backslash j} [b^{\nu}_{k\rightarrow i}] + H^{\nu}_{ji} - \min \bigg[0, \min_{\substack{k,l\in\partial i\backslash j\\k\neq l}} [a^{\nu}_{k\rightarrow i}+b^{\nu}_{l \rightarrow i}] \bigg], & \text{if } i\not \in \{s^{\nu},e^{\nu}\} \\ 
         \infty, & \text{if } i=s^{\nu} \\ \\
         H^{\nu}_{je^{\nu}}-\min_{k\in \partial e^{\nu} \backslash j } [a^{\nu}_{k\rightarrow e^{\nu}}], & \text{if } i=e^{\nu}
        \end{array}\right.  \\ \nonumber 
\end{align}
From the messages the flow is computed as $I^{\backslash\nu}_{ij} = \sum_{\mu\in [M] \backslash\nu} (\sigma^{\mu}_{i\rightarrow j} + \sigma^{\mu}_{j\rightarrow i})$, with

\begin{align}
  \sigma^{\nu}_{i\rightarrow j}  = & \I[i=s^{\nu}] \Theta\bigg(
  \min_{l\in\partial j\backslash i} (b^{\nu}_{l\rightarrow j})-b^{\nu}_{i \rightarrow j} 
  \bigg) \nonumber \\ &+ (1-\I[i=s^{\nu}])\Theta\bigg( \min{\bigg( 0,\min_{\substack{l\in \partial j \backslash i}}{(a^{\nu}_{l\rightarrow j } + \min_{\substack{k \in \partial j \backslash l} } b^{\nu}_{k\rightarrow j})} \bigg)} -(a^{\nu}_{i \rightarrow j}+ \min_{\substack{k \in \partial j\backslash i}} b^{\nu}_{k\rightarrow j}) \bigg).
\label{eq:marginal_directed}
\end{align}
In \cite{yeung2013physics} the quantity $H^{\nu}_{ij} = \phi(I_{ij}^{\backslash\nu}+1)-\phi(I_{ij}^{\backslash\nu})$ is Taylor approximated with $\phi'(I_{ij}^{\backslash\nu})$.
The CBP algorithm alternates between updating all messages once using \eqref{eq:BP_a}, \eqref{eq:BP_b} and then updating the flows using \eqref{eq:marginal_directed}. This is repeated until convergence is reached, at which point, the paths are determined by plugging the final values of the messages into equation \eqref{eq:marginal_directed}. For our experiments we use the original implementation of CBP from \cite{yeung2013physics}. Its pseudocode, presenting some minor differences with what we derived, is provided in Appendix \ref{app:BP_derivation}.

By picking a different update schedule one can essentially retrieve the greedy algorithm. Suppose we pick a path $\nu$ and update the messages $\{a_{i\to j}^\nu,b_{i\to j}^\nu\}_{(i\to j) \in E}$ several times until they reach convergence. Then BP approximates the conditional $P_\infty\left(\pi^\nu|\{\pi^\mu\}_{\mu\in[M]\backslash\nu},\{s^\mu,t^\mu\}_{\mu\in[M]}\right)$ which is supported on shortest paths from $s^\nu$ to $t^\nu$. Therefore, with this schedule, BP approximates the shortest path computation described in \ref{sec:greedy_algo}. If one then updates the variables $\{\sigma_{i\to j}^\nu\}_{(i\to j)\in E}$ using \eqref{marginal_directed} and then repeats the whole procedure for other values of $\nu$ this reproduces the greedy algorithm.

\subsection{Simulated annealing}
Given an energy function $H$, to be minimized, the simulated annealing technique consists of building a probability measure  $\propto e^{-\beta H}$ and then drawing samples from it while increasing $\beta$. The sampling part is usually accomplished using iterative methods such as MCMC.
For not-too-high values of $\beta$, the noise in the distribution will allow the dynamics to jump possible barriers in $H$ and explore the landscape. When increasing $\beta$, the measure will become more and more concentrated on the global minimum of $H$. In practice $\beta$ is changed with the iteration number $t$, according to a schedule $\beta(t)$.
For ITAP, the simulated annealing aims to draw samples from the ($t$ dependent) probability measure $P_{\beta(t)}\left(\{\pi^\mu\}_{\mu\in[M]}|\{s^\mu,t^\mu\}_{\mu\in[M]}\right)$, where $\beta(t)$ is an increasing function. To accomplish this task we use a MCMC based on Gibbs sampling. Gibbs sampling \cite{geman1984stochastic} is an iterative MCMC algorithm that samples from the joint distribution \eqref{eq:gibbs_measure} by sampling each variable from its conditional distribution \eqref{eq:conditional_gibbs} in a sequence. The details of annealing schedule and the pseudocode for the algorithm are provided in Appendix \ref{app:sim_annealing}. Here we explain how it works using a mock example with three paths.
To lighten the notation define the joint measure $P_\beta(\pi^1,\pi^2,\pi^3)\coloneqq P_\beta\left(\{\pi^\mu\}_{\mu\in[3]}|\{s^\mu,t^\mu\}_{\mu\in[3]}\right)$, and the conditional measure $P(\pi^1|\pi^2,\pi^3)\coloneqq P_\beta\left(\pi^1|\{\pi^\mu\}_{\mu\in[3]\backslash1},\{s^\mu,t^\mu\}_{\mu\in[3]}\right)$. Initially, given the configuration $\pi^1(t),\pi^2(t),\pi^3(t)$, the algorithm samples $\pi^1(t+1)$ from $P_{\beta(t)}(\pi^1|\pi^2(t),\pi^3(t))$, followed by sampling $\pi^2(t+1)$ from $P_{\beta(t)}(\pi^2|\pi^1(t+1),\pi^3(t))$, and finally $\pi^3(t+1)$ from $P_{\beta(t)}(\pi^3|\pi^1(t+1),\pi^2(t+1))$. If $\beta(t)=\beta$ were constant in $t$, then repeating this process, the samples $(\pi^1(t),\pi^2(t),\pi^3(t))$ would converge to the distribution $P_\beta(\pi^1,\pi^2,\pi^3)$ \cite{casella1992explaining,robert1999monte}. In the simulated annealing, since $\beta(t)$ increases towards infinity, there is no stationary distribution and the above stochastic process will lead to a local minimum of $H$.

How do we sample from the conditional distribution \eqref{eq:conditional_gibbs}? The conditional is a probability measure over a single self-avoiding path with constrained endpoints and subject to the Gibbs measure $\propto e^{-\beta\Delta H }$. The self-avoiding constraint makes this measure hard to sample directly. Instead, we resort to a second, Metropolis based, MCMC that allows us to sample self-avoiding paths. The method, thoroughly described in appendix \ref{app:SAW_sampler}, consists of generating SAWs using a proposal distribution which is easy to sample and that approximates the conditional. Then a Metropolis accept/reject step is carried out. The proposal distribution is chosen to capture the leading order behavior of the conditional when $\beta\to\infty$. This means that the sampler is particularly efficient (low rejection rate) at high $\beta$. 

\subsection{Relaxed ITAP solver (RITAP)}
\label{sec:TAP_relaxation}
The last algorithm we study is based on relaxing ITAP to TAP, solving TAP, and finally projecting the solution of TAP to a solution of ITAP. When having to solve TAP we will assume that we have an algorithm \verb|TAPsolver|$(G,D)$, that given a graph $G=(V,E)$ and a demand matrix $D$, returns a collection of sets $\PM=\{\PM_{xy}, (x,y)\in (V\times V)\}$. 
Every set  $\PM_{xy}$ contains the paths from $x$ to $y$ that at a local minimum (in the space of edge flows, see \ref{def:tap_local_min}) of TAP carry some positive flow, as well as the flows carried by each path. Mathematically we have
\begin{equation}
\label{eq:collection_TAP}
    \PM_{xy}=\left\{(\pi^a_{xy},h^a_{xy}),\;(x,y)\in(V\times V),\, a\in\{1,\dots,R_{xy}\} \right\}.
\end{equation}
We indicate with $\PM_{xy}.\pi$ the set of paths $\{\pi_{xy}^a\}_{a=1}^{R_{xy}}$, and with $\PM_{xy}.h$ the set of path flows $\{h_{xy}^a\}_{a=1}^{R_{xy}}$. $R_{xy}$ is the number of paths between $x$ and $y$ carrying a positive flow in the solution.
In the repulsive case, thanks to convexity, there is only one (global) minimum in the space of edge flows. Therefore \verb|TAPsolver| outputs some global optimal path flows corresponding to the unique optimal edge flows.\footnote{Recall that while the global minimum is unique in the space of edge flows, it is not in the space of path flows. The particular path flows outputted depend on the implementation of TAPsolver} In contrast TAP is NP-hard in the attractive case; the particular local minimum reached by \verb|TAPsolver| depends on the implementation. Our implementation is based on the Frank-Wolfe algorithm. This is neither the sole nor the fastest algorithm to solve TAP \cite{perederieieva2015framework}, however we use it due to its simplicity. 

To obtain a valid solution to the integer problem we devise a procedure to project the relaxed solution to an integer one. In the attractive case, proposition \ref{prop:no_flow_splitting_concave} guarantees that 
\verb|TAPsolver| already outputs integer paths, thus the projection step is superfluous.
In the repulsive case we use the following procedure.
Looping over $\mu\in[M]$ we consider $\PM_{s^\mu t^\mu}$. We set $\pi^\mu$ (the $\mu$th integer path) to be the path in $\PM_{s^\mu t^\mu}$ that carries the most flow: mathematically $a^\star= \argmax_a h^a_{s^\mu t^\mu}$ and $\pi^\mu=\pi_{xy}^{a^\star}$. After assigning $\pi^{\mu}$ we update $h^{a^\star}_{s^\mu t^\mu}\gets h^{a^\star}_{s^\mu t^\mu}-1$, this way, if there are multiple paths with the same origin and destination, one takes into account the flow that is already assigned. We report below the pseudocode describing this procedure. Since it contains a loop over $M$ the complexity of this algorithm is at least linear in $M$. In appendix \ref{app:ITAP_relaxed} we present a smarter implementation of RITAP whose time complexity depends on $M$ as $\min(M,N^2)$, and is therefore more suited to very high $M$. The same appendix also contains additional details about the Frank-Wolfe algorithm. 
\begin{algorithm}[H]
\caption{Relaxed ITAP solver (RITAP)} 
\label{alg:relaxed_ITAP_solver}
\begin{algorithmic}[1]
\State\textbf{Input: }  $G$ graph over which paths are sampled, $D$ demand matrix. 
\State\textbf{Output:} $\{\pi^\mu\}_{\mu\in[M]}$ the paths outputted by the algorithm
\State $\PM \gets \verb|solveTAP|(G,D)$. \Comment{Solving TAP}
\For{$\mu=1,\dots,M$} \Comment{Loop over OD pairs}
\State $a^\star\gets \argmax_a \PM_{s^\mu t^\mu}.h$ \Comment{Find path with maximum flow}
\State $\pi^\mu\gets \PM_{s^\mu t^\mu}.\pi [a^\star]$ 
\State $\PM_{s^\mu t^\mu}.h[a^\star]\gets \PM_{s^\mu t^\mu}.h[a^\star]-1$\Comment{Update the flow on $\pi_{s^\mu t^\mu}^{a^\star}$}
\EndFor
\end{algorithmic}
\end{algorithm}

\section{Results on random instances}
\label{sec:results_random_model}
To test the algorithms and study ITAP, we consider it on undirected sparse random regular graphs. Origin-destination pairs are sampled i.i.d. from the uniform distribution over pairs of distinct vertices. Mathematically we sample $M$ times independently from $P(x,y)=\mathbb{I}[x\neq y]/(N(N-1))$, where $x,y \in (V\times V)$. 
We consider nonlinearities of the form $\phi(x)=x^\gamma$, with $\gamma>0$. In the repulsive case we set $\gamma=2$, while in the attractive case $\gamma=1/2$.
We run extensive numerical experiments with different algorithms to study their ability to minimize the ITAP energy $H(\pmb I)$. Additional details about the experimental settings are given in appendix \ref{app:num_exp_details}.

Given a graph $G=(V,E)$ and a set of OD pairs $\{s^\nu,t^\nu\}_{\nu\in[M]}$, each algorithm outputs a set of paths $\{\pi^\nu\}_{\nu\in[M]}$ connecting the respective OD pairs. We name $H$ the energy reached by the algorithm. To establish a baseline, we consider a very basic algorithm that, for each OD pair $(x,y)$, outputs the shortest path (in the topological distance on $G$) between $x$ and $y$. We call the resulting energy $H_\text{SP}$.\footnote{When the shortest path is not unique we output one of them; this arbitrariness does not significantly affect our results.} We use the shortest paths as a benchmark since they represent the optimal routing for the non-interacting case (every path is routed independently of others), corresponding to $\gamma=1$. In fact in this case the energy can be written as $H(\pmb I)=\sum_{\mu\in[M]}\sum_{e\in\pi^\mu} 1=\sum_{\mu\in[M]}\text{Length}(\pi^\mu)$. Comparing $H$ to $H_\text{SP}$ also allows us to gauge the effective strength of interactions, that is how much gain in energy one achieves by taking the interaction into account.

\subsection{Regime with the average flow per edge $O(1)$}
\label{section_results_comparison_algos_small_M}

For every considered algorithm we plot the fraction of energy saved when compared to the shortest paths, mathematically expressed as $1-H/H_\text{SP}$.
To minimize the energy we employ several algorithms: Greedy, CBP, simulated annealing, and RITAP. In each curve we fix the degree $d$ and the number of nodes $N$, and vary the number of OD pairs $M$. We plot quantities as a function of the effective parameter \begin{equation}
\label{eq:rho}
    \rho\coloneqq \frac{2M\log N}{Nd\log d}.
\end{equation} This will allow us to compare different sizes or degrees. The choice of this scaling is justified both by the empirical collapse of the curves with different $N,d$ and by the fact that $\rho$ is approximately the average the flow per edge therefore it is a proxy for how congested the graph is. To see this consider the shortest path algorithm. In an RRG on average the shortest path between two nodes has an average length of approximately $\log N/\log d$ \cite{tishby2022mean}. Then the total flow $\sum_{e\in E} I_e$ present on the graph is about $M\log N/\log d$, while the number of edges is $Nd/2$. This gives a flow per edge of $\rho$. While the average length of the outputted paths changes with the algorithm, the difference is small enough not to compromise the validity of the approximation. The agreement between $\rho$ and the flow per edge is further verified with numerical experiments in Appendix \ref{app:num_exp_details}. To compare algorithms, we focus initially on the regime where $\rho$ is of order one, since for higher $\rho$, most algorithms have similar performances, and for $\rho\to 0$, the system is non-interacting. Furthermore in most cases, the quantity $1-H/H_\text{SP}$ peaks for $\rho\approx 1$, indicating that this is where the effective coupling among paths is the greatest. Error bars in all plots correspond to a $\pm 1$ standard deviation interval around the mean.
\subsubsection{Attractive paths case ($\gamma=1/2$)}
We first consider the case in which paths attract each other, corresponding to a concave $\phi$. In this case, TAP local optima coincide with ITAP ones, and both problems are NP-hard. We start by comparing the algorithms for $N=500$ and $d=3$, as portrayed in figure \ref{fig:allalgos_gamma05}. 
\begin{figure}[!ht]
     \centering
     \includegraphics[scale=1]{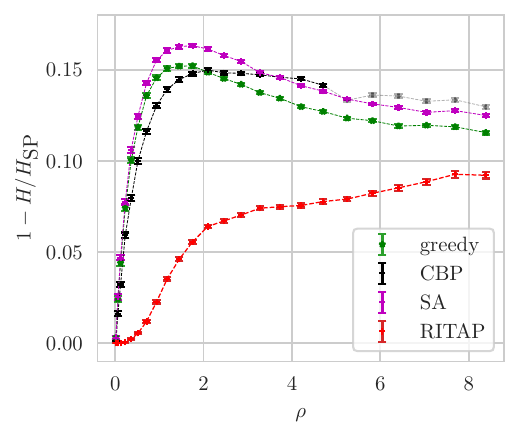}
     \caption{Relative energy difference between the paths obtained using each algorithm and the shortest paths. The higher the better. Algorithms used are: greedy, CBP, simulated annealing, and RITAP. The CBP line is semi-transparent when CBP converges on less than 80\% of the instances. Results are averages over 200 instances of ITAP, with $d=3, \,N=500, \,\gamma=1/2$.}
     \label{fig:allalgos_gamma05}
\end{figure}
We immediately remark that RITAP is less effective than other algorithms. This is not surprising as in the attractive case, the cost function to minimize is highly non-convex, and RITAP gets easily stuck in a local minimum.
Simulated annealing is observed to be the best-performing algorithm in this case. CBP instead is affected by convergence issues at high $\rho$. 
When it converges for less than 80\% of the instances we use a semi-transparent line in the plots. For high $\rho$, averaging over converged instances, CBP achieves marginally better performance compared to simulated annealing.

To consolidate our findings, we repeat the experiments for different values on $N,d$. Figure \ref{fig:gamma05_varyN} shows the curves of energy for a fixed $d=3$ and for several graph sizes.
\begin{figure}
    \centering
    \includegraphics[scale=1]{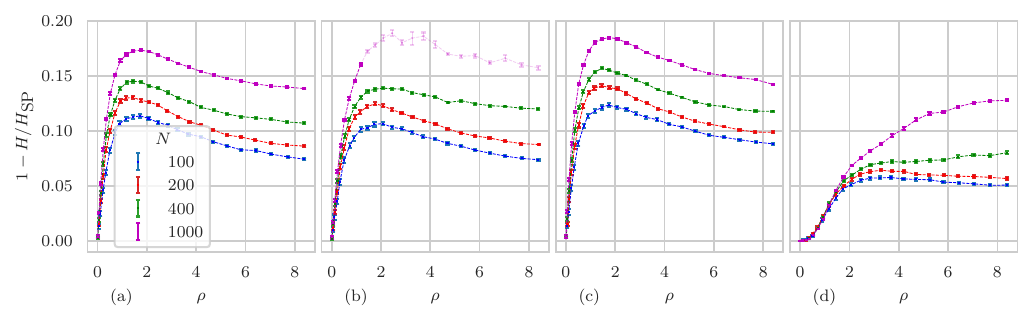}
    \caption{Relative energy difference between the paths obtained using each algorithm and the shortest paths. The higher the better. The algorithms used are: (a) greedy, (b) CBP, (c) simulated annealing, and (d) RITAP. The CBP line is semi-transparent when CBP converges on less than 80\% of the instances. We fix $\gamma = 1/2$, and consider graphs of degree $d=3$ with different sizes $N\in\{100,200,400,1000\}$. Results are averages over 200 instances of ITAP.}
    \label{fig:gamma05_varyN}
\end{figure}
Notice that the larger the size, the more energy is saved compared to the shortest paths. The peaks of curves with different $N$ align, confirming the validity of $\rho$ as a size-independent measure of congestion.
We remark that for higher $N$, CBP becomes less stable, but at high $\rho$, if it converges, it is usually the best-performing algorithm. 
Apart from CBP, simulated annealing maintains a small but consistent margin over other algorithms.
In the last set of experiments, whose results are depicted in figure \ref{fig:gamma05_varyd}, we set $N=200, \gamma=1/2$ and vary $d$. 

\begin{figure}
    \centering
    \includegraphics[scale=1]{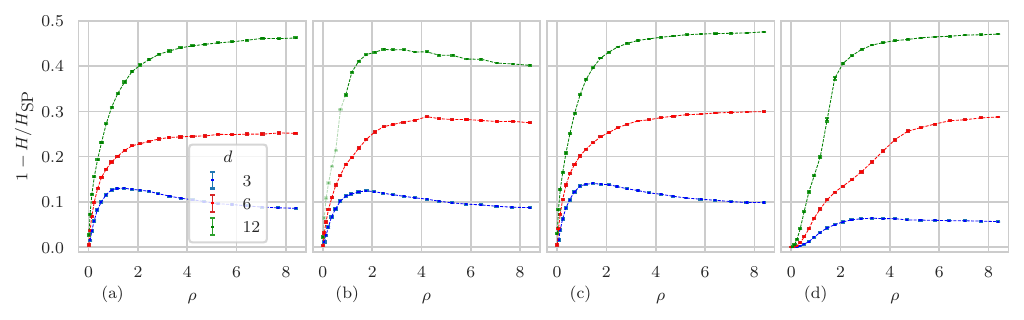}
    \caption{Relative energy difference between the paths obtained using each algorithm and the shortest paths. The higher, the better. The algorithms used are (a) greedy algorithm, (b) CBP, (c) simulated annealing, and (d) RITAP. Transparent parts of a line correspond to values for which the algorithm converges only for less than 80\% of instances. Results are averages over 200 instances of ITAP, with $N=200,\gamma=1/2$ and $d\in\{ 3,6,12\}$.}
    \label{fig:gamma05_varyd}
\end{figure}

We remark that the higher the degree of the RRG, the higher will be the energy saved when comparing to the shortest paths. We also observe that the performance gap between RITAP and other algorithms closes at high $d$ and $M$. For higher $d$, the proportion of instances for which CBP converges changes haphazardly with $M$. For instance, for $N=200$ and $d=12$, the proportion is above 80\% for the two first points, then decreases to around 60\% (transparent part) and goes back to higher values (opaque part). 

\subsubsection{Repulsive paths case ($\gamma=2$)} 
We now turn to the repulsive case, where $\phi$ is convex. The relaxed problem, TAP, in this case is a convex, and hence easy, optimization problem. 
Figure \ref{fig:allalgos_gamma2} shows the dependence of $1-H/H_\text{SP}$ on $\rho$ for different algorithms. The degree and the size of the graphs are fixed to $d=3$ and $N=500$. We also plot the curve corresponding to the TAP optimal energy, which constitutes an upper bound to the energy saving that can be achieved by any algorithm that minimizes the energy under the ITAP constraints. 
\begin{figure}[!ht]
     \centering
     \includegraphics[scale=1]{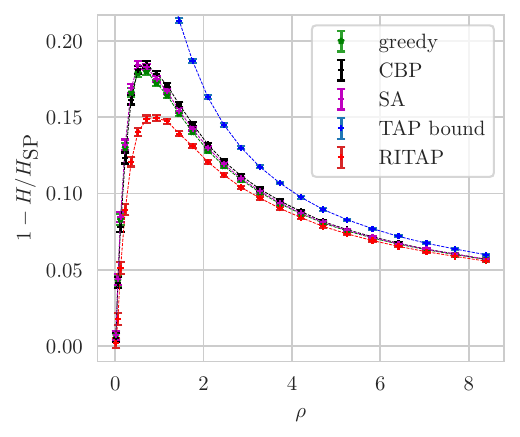}
     \cprotect\caption{Relative energy difference between the paths obtained using each algorithm and the shortest paths. The higher the better. Algorithms used are: greedy, CBP, simulated annealing, RITAP and, as an upper bound, the TAP optimal energy. Results are averaged over 200 instances of ITAP, with $d=3, \,N=500, \,\gamma=2$.}
     \label{fig:allalgos_gamma2}
\end{figure}
Greedy, CBP, and simulated annealing display approximately the same performance in terms of saved energy, instead RITAP is slightly less efficient. Increasing $\rho$ however all algorithms perform comparably and the gap between saved energy in ITAP and the TAP bound progressively decreases. Considering very small $M$, all the curves (except for the bound) go to zero. This is because when only a few paths are present, they are unlikely to share edges and therefore do not interact. In absence of interactions, the shortest paths are the global optimal solution. Conversely, the TAP bound follows a different behavior. For small values of $\rho$, the quantity $1-H/H_\text{SP}$ is close to 0.8. 
Indeed, since the paths are not constrained to be integer, the flow is split across several paths, allowing a considerable advantage of energy compared to any integer configuration. 

 In figure \ref{fig:gamma2_varyN} we vary $N$ and keep $d=3$ fixed. We observe that the size does not modify the phenomenology observed in \ref{fig:allalgos_gamma2}. Moreover the curves for different $N$ overlap with each other.
\begin{figure}[!ht]
    \centering
    \includegraphics[scale=1]{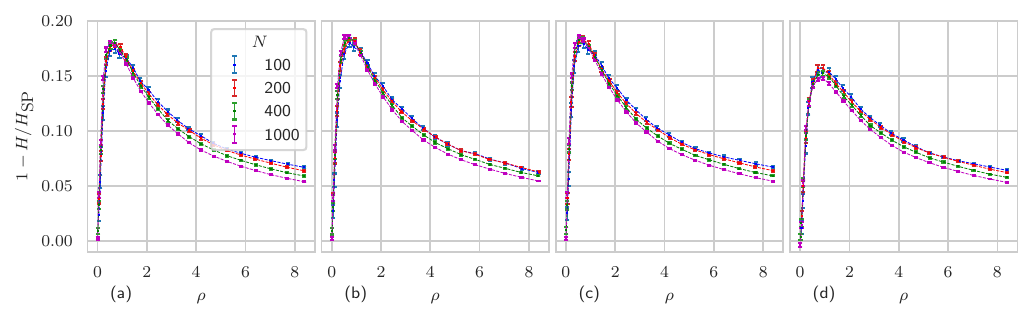}
    \caption{Relative energy difference between the paths obtained using each algorithm and the shortest paths. The algorithms used are: (a) greedy , (b) CBP, (c) simulated annealing, and (d) RITAP. The higher the better. Graphs of degree $d=3$ with different sizes $N\in\{100,200,400,1000\}$ are used. Results shown are averages over 200 realizations with $\gamma=2$.}
    \label{fig:gamma2_varyN}
\end{figure}
Next we fix $N=200,\gamma=2$ and vary the degree of the RRG. Figure \ref{fig:gamma2_varyd} illustrates the results. Similarly to the $\gamma=1/2$ case, as the degree increases, the quantity $1-H/H_\text{SP}$ reaches larger values, indicating that the shortest path algorithm is increasingly suboptimal.
Additionally, the simulated annealing gains some slight advantage over other algorithms for higher $d$. In the same figure we also plot the TAP upper bound; while for $d=3,6$ the gap between TAP and ITAP efficiencies narrows when increasing $\rho$, it appears to widen for $d=12$. However, for even higher $\rho$, the gap decreases and ultimately approaches zero (see figure \ref{fig:TAP_ITAP_conv_N128_varyd}).
Results for CBP have not been reported for $d=6,12$ since respectively for $\rho\geq 0.43$ and $\rho\geq 0.14$, more than half of the 200 instances fail to converge. Last, we remark that $\rho$ yields again a good collapse of the curves with different degrees, by aligning their peaks.

We end this section with a note on the speed of different algorithms.
In general the greedy algorithm stands out as the most computationally efficient algorithm that achieves good performance over the considered range of $M$. RITAP is also fast computationally, but not as effective unless $\rho$ is large. All the other algorithms incur in higher computational costs, posing challenges to scalability for large sizes.
\begin{figure}[!ht]
    \centering
    \includegraphics[scale=1]{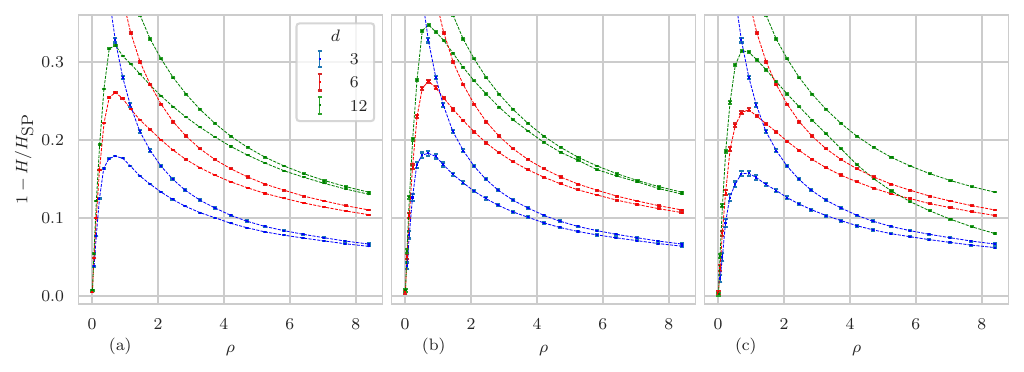}
    \caption{Relative energy difference between the paths obtained with each algorithm and the shortest paths. Lower curves are ITAP solutions outputted by: (a) greedy algorithm, (b) simulated annealing, (c) RITAP. The higher the better. The upper curves represent the TAP bound. Graphs of size $N=200$ with $\gamma=2$ and different degrees $d\in\{3,6,12\}$ are used. The results shown are averages over 200 realizations.}
    \label{fig:gamma2_varyd}
\end{figure}

\subsection{Regime of large average flow per edge $M = O(N^2)$: Integrality of relaxed solutions}
\label{sec:integ_relax_solutions}
In the previous section we saw that the parameter $\rho$ is the relevant quantity to consider when describing systems where the average flow per edge of order one. Here we argue that the problem exhibits a second intrinsic scaling regime when $M=\Theta(N^2)$.
We then investigate whether TAP solutions are also ITAP solutions when $M$ is large and $\phi$ is convex. Defining an effective parameter \begin{equation}
\label{eq:eta}
    \eta\coloneqq \frac{M}{N(N-1)},
\end{equation} we observe that a stable asymptotic behavior arises for $\eta\approx 1$ and continues for higher $M$.

We have $\eta=\E_D[D_{xy}]$, hence this regime corresponds to a dense demand matrix. 
We run RITAP for several different values of $d, N$ and plot quantities as a function of $\eta$. In this section we report the results for fixed $d=3, \gamma=2$ and variable $N$. Appendix \ref{app:num_exp_M=N^2} illustrates the additional results when varying $d$ and when using more realistic nonlinearities.
\begin{figure}[]
    \centering
    \includegraphics[scale=1]{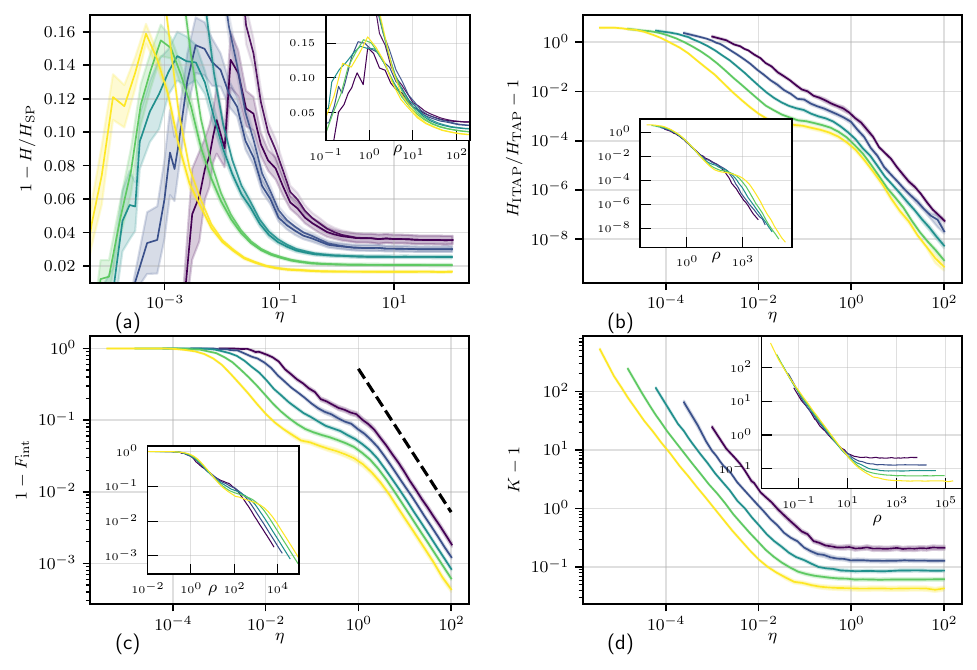}
    \caption{Properties of the TAP and ITAP solutions for large $M$.
    Experiments with RRG with $d=3$, $\gamma=2$. From darker to lighter the curves correspond to $N=32,64,128,256,512$. The shaded area indicates a $\pm1$ standard deviation interval around the mean. In insets, the same quantities as the main plots are plotted as a function of $\rho$. (a) Efficiency with respect to shortest paths for TAP (upper curves) and ITAP (lower curves) as a function of $\eta$. (b) Dependence on $\rho$ of the ratio between the TAP optimal energy and the energy obtained after projecting the TAP solution onto an integer one using RITAP. (c)  Fraction of non-integer paths in TAP as a function of $\eta$. The dashed line has a slope of $1/\eta$ (equivalently $1/M$), to illustrate asymptotic behavior. (d) Average degeneracy of support paths as a function of $\eta$.}
    \label{fig:energy_conv_tap_itap}
\end{figure}
Panel (a) in figure \ref{fig:energy_conv_tap_itap} shows the quantity $1-H/H_{SP}$ (averaged over instances) as a function of $\eta$, for several $N$s. 
In the upper curves $H$ is the optimal energy in TAP (let us call it $H_\text{TAP}$). In the lower curves it is the energy of the ITAP solution obtained with RITAP; we shall name this quantity $H_\text{ITAP}$. We notice that for low $M$ the two curves take quite different values, however as $M$ increases the gap between the two shrinks towards zero. To closely examine the convergence we also plot the ratio $H_\text{ITAP}/H_\text{TAP}-1$ in the right panel. The fact that this quantity goes to zero confirms that $H_\text{ITAP}\to H_\text{TAP}$ when $\eta\to \infty$. Since $H_\text{TAP}$ is a lower bound to $H_\text{ITAP}$, this result shows that RITAP is basically reaching the global minimum of ITAP, and thus in the high $\eta$ regime, ITAP is easy to solve approximately.

Next we explore how the convergence in energy reflects on the paths through which flow is routed in TAP and ITAP solutions. We say a path is integer if it carries an integer amount of flow at optimality in TAP. Using the notation of section \ref{sec:TAP_relaxation}, consider a TAP solution $\PM$. We define the fraction of integer paths in TAP as $F_\text{int}\coloneqq\frac{1}{M}\sum_{x,y \in (V\times V)} \sum_{a\in R_{xy}} \lfloor h_{xy}^a\rfloor$. \footnote{$\lfloor h_{xy}^a\rfloor$ is indeed the integer part of the flow being routed on $\pi_{xy}^a$.}
In other words, $F_\text{int}$ is the fraction of TAP optimal paths that already satisfy ITAP constraints. 
Figure \ref{fig:energy_conv_tap_itap}(c) depicts the dependence of $1-F_\text{int}$ on $\eta$. $1-F_\text{int}$ decreases as $\eta$ increases, moreover when $\eta\approx 1$, we observe the onset of a power law behavior $1/\eta$ as indicated by the dashed black line.
Overall this shows that also on the path level, for $\eta\gg 1$, a solution to TAP is almost a solution to ITAP, in the sense that almost all TAP paths are integer ($F_\text{int}\to 1$).
 In Appendix \ref{app:proofs} we prove that $F_\text{int}\to 1$ implies that $ H_\text{ITAP}/H_\text{TAP}\to 1$, therefore rigorously establishing the relation between path integrality and energy difference in TAP and ITAP.
 
 The main drawback of studying TAP paths at optimality is that these are not unique, thus we could be observing some peculiarity of the Frank-Wolfe algorithm rather than a fundamental property of the problem. To erase this doubt and to cast light onto the behavior of $F_\text{int}$ we consider a quantity which is unique at optimality in TAP. 
We start from the observation that, although not unique, path flows at optimality are not arbitrary either. Based on the optimal edge flows one can determine a set of paths on which positive flow is possible. We call these the support paths. The following result by Wardrop \cite{doi:10.1680/ipeds.1952.11259} characterizes the support paths.
\begin{proposition}[Wardrop's second principle]
\label{prop:wardrop_2_principle}
Let $G,D$ be a TAP instance and $\{h_{xy}^{\star a}\}_{(x,y)\in V\times V, a\in [|\Pi_{xy}|]}$ be any path flows at optimality. Also let $\pmb{I^\star}$ be the edge flows at optimality. \textbf{Then} for every $(x,y)\in V\times V$, for every $a\in[|\Pi_{xy}|]$, a necessary condition for $h_{xy}^{\star a}$ to be strictly positive is that
\begin{equation}
    \sum_{e\in \pi_{xy}^a} \phi'(I_e^\star)=\min_{\pi\in\Pi_{xy}}  \sum_{e\in \pi} \phi'(I_e^\star)
\end{equation}
\end{proposition}
A proof of this fact is presented in Appendix \ref{app:proofs}.
Wardrop's second principle says that the support paths are the shortest paths in the graph whose edges $e\in E$ are weighed with $\phi'(I_e^\star)$. If the optimal flow between nodes $x$ and $y$ is split across several paths, it means all of these paths are the shortest paths in the aforementioned graph. We indicate with $\mathcal S_{xy}(\pmb{I^\star})\coloneqq\argmin_{\pi\in\Pi_{xy}} \sum_{e\in \pi} \phi'(I_e^\star)$ the set of support paths starting and terminating respectively at nodes $x$ and $y$. Notice that for a given instance of TAP, the support paths are unique, as they depend only on the optimal edge flows.  The larger the cardinality of $\mathcal S_{xy}(\pmb{I^\star})$, the more potential split across paths, the less integrality in TAP paths.
Let $Q\coloneqq\sum_{x,y\in V\times V} \mathbb I[D_{xy}>0]$ be the number of nonzero entries in the demand matrix.
We define the average TAP path degeneracy to be $K\coloneqq \frac{1}{Q}\sum_{x,y\in V\times V} \mathbb I[D_{xy}>0]\left|\mathcal S_{xy}(\pmb{I^\star})\right|$. 
The behavior of $K-1$ as a function on $\eta$ is plotted in figure \ref{fig:energy_conv_tap_itap} (d). For small $\eta$, there is a huge degeneracy in the support paths, with the flow being possibly split across hundreds of paths. When $\eta$ increases the degeneracy decreases and plateaus to a fixed value when $\eta$ becomes larger than one.

Let us explain why the degeneracy plateauing plays a crucial role in guaranteeing that TAP converges to ITAP. Fix two nodes $x,y$ and suppose $\left|\mathcal S_{xy}(\pmb{I^\star})\right|$ is constant in $M$: when $M$ increases so does $D_{xy}$, however the flow will always be routed through the same number of paths joining $x$ and $y$. In turn this implies that an increasing fraction of the flow is integer, resulting in TAP solutions being close to an ITAP one.\footnote{a moment of thought reveals that the fractional part of the flow is smaller than $\left|\mathcal S_{xy}(\pmb{I^\star})\right|$, which is constant in $M$. Therefore the proportion of fractional flow is at most $\left|\mathcal S_{xy}(\pmb{I^\star})\right|/D_{xy}$ which goes to zero.} Translating this reasoning in formulas we arrive at the following identity which gives a lower bound to $F_\text{int}$ in terms of the paths degeneracies.
\begin{equation}
\label{eq:bound_Fint_degen}
     F_\text{int}\geq1-\frac{1}{\eta}\left(K-1\right)
\end{equation}
We give a proof of it in Appendix \ref{app:ITAP_relaxed}. For example if $S_{xy}(\pmb{I^\star})=1$ (no degeneracy) for all $(x,y)$, then the bound gives $F_\text{int}\geq 1$.
Therefore if $K$ is constant in $M$, $1-F_\text{int}$ goes to zero like $1/\eta$ if not faster, reproducing the rate observed in figure \ref{fig:energy_conv_tap_itap} (c). 

One question stands: why when $M$ increases the degeneracy decreases and ultimately plateaus? We conjecture that this phenomenon is related to the demand matrix becoming more dense and uniform (in the sense that all entries are of the same order) for higher $M$. In particular, the distribution of OD pairs, rather than their number, plays a crucial role in reducing the degeneracy. We support this claim by showing that the number of OD pairs can be increased while leaving the degeneracy unchanged. Consider two TAP instances with the same $G$ and demand matrices respectively $D$ and $2D$. Let $\phi(x)=x^\gamma$ with $\gamma\geq 1$, and suppose we know the path flows at TAP optimality in the instance $G, D$. Then we show in Appendix \ref{app:proofs} that by keeping the paths the same and doubling the flow on each of them, we obtain a TAP optimal solution to the instance $G,2D$. This implies that the support paths do not change. We conclude that the degeneracy is not influenced by the norm of $D$. 
The above results and discussion lead us to formulate the following conjecture. 
\begin{conjecture}
\label{conj:limit_degen}
    Let $G\sim \text{RRG}(N,d)$ be an RRG with $N$ nodes and degree $d$, and $D$ be a demand matrix obtained by sampling $M$ OD pairs uniformly at random. Suppose $\phi$ is convex and let $K(G,D)$ be the average degeneracy in the TAP instance $G,D$. Define  
    $\kappa(\eta)\coloneqq \lim\limits_{\substack{M,N\to\infty\\ M/(N(N-1))=\eta}} \E_{G,D} K(G,D)$. \textbf{Then} $\lim\limits_{\eta\to\infty} \kappa(\eta)$ exist and is finite.
\end{conjecture}
Plugging this result into $F_\text{int}$ one sees that $ \lim\limits_{\substack{M,N\to\infty\\ M/(N(N-1))=\eta}} \E_{G,D}F_\text{int}$, approaches $1$ as $\eta\to\infty$.

\section{Two scales present also on real-world graphs}

\subsection{A tale of two scales}
In our above analysis, we identified two main regimes in which the problem's behavior is interesting. These correspond respectively to when the effective flow per edge $\rho$ and $\eta$, defined in \eqref{eq:rho} and \eqref{eq:eta}, are of order one. Figure \ref{fig:energy_conv_tap_itap} illustrates the presence of the two regimes by showing the same quantities plotted as a function of $\eta$ and $\rho$, respectively, in the main panels and in the insets. One observes that when $\rho=\Theta(1)$ the curves for different sizes collapse when plotted against $\rho$. The importance of the regime $\rho=\Theta(1)$ is that this is where the effective interactions among paths arise and are the strongest. Panel (a) shows that the greatest relative difference between the ITAP energy and the shortest paths (non-interacting case) one is indeed attained for $\rho\approx 1$. Instead for $\rho\to0$ the effective interaction goes to zero and the ITAP energy approaches that of the shortest paths. \footnote{the fact that the peak of $1-H/H_\text{SP}$ happens for $\rho\approx 1$ is found also in all other plots, except for figure \ref{fig:gamma05_varyd}.}

Increasing $M$ we move into a setting there $\rho\gg 1$ and $\eta\approx 1$. In panels (b,c) one observes that the 'knee' of the curve, where the power law arises, indeed corresponds to $\eta\approx 1$, but different values of $\rho$, depending on the size.
Similarly, in panel (d), the degeneracy plateaus at fixed $\eta$ but not at fixed $\rho$. 

\subsection{Validation on real-world graphs}

\label{sec:real_data_exp}
To prove the robustness of our results we consider two real-world traffic networks:
\begin{itemize}
    \item The Anaheim city road network, composed by $N=416$ nodes. The edges are directed and average total degree (in degree plus out degree) is 4.39. 
    \item The Eastern Massachusetts highway network, with $N=74$ nodes. The edges are directed and average total degree is 6.97.
\end{itemize}
Both networks are taken from \cite{bstabler}. We use $\phi(x)=x^\gamma$ as nonlinearity and sample the OD pairs uniformly at random, always with the constraint that the origin and destination must be distinct. 
Figure \ref{fig:real_nets_highM} depicts the results of these experiments. 
\begin{figure}
    \centering
    \includegraphics[scale=1]{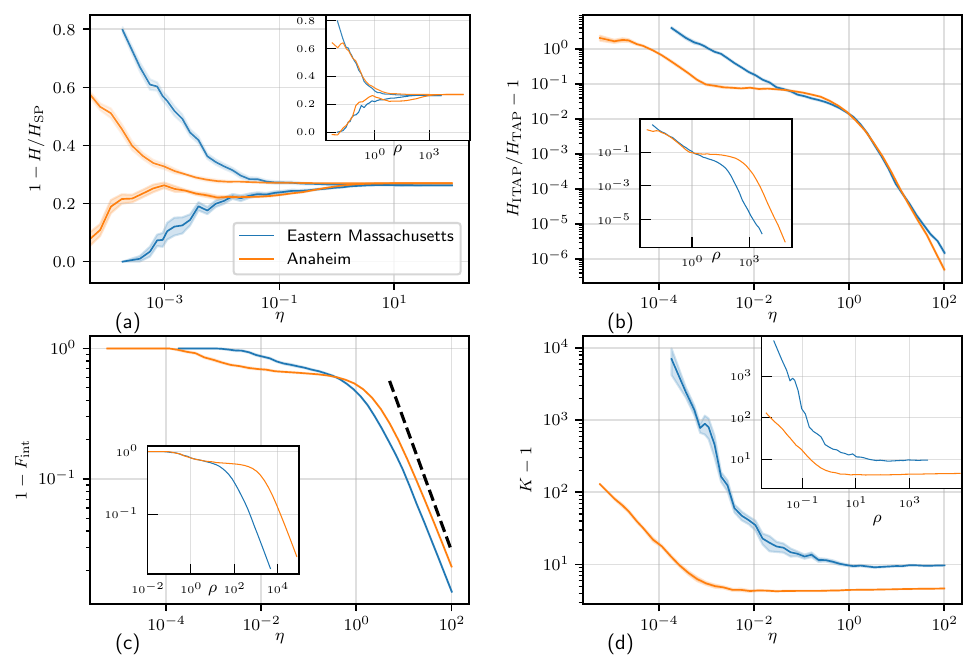}
    \caption{Experiments on Anaheim and Eastern Massachusetts networks with OD pairs picked uniformly at random, and nonlinearity $\phi(x)=x^2$. The shaded area indicates a $\pm1$ standard deviation interval around the mean. The randomness is over 20 realizations of the OD pairs. In insets, the same quantities as the main plots are plotted as a function of $\rho$. (a) Relative energy gain with respect to the shortest path routing. The upper curves correspond to TAP optimality, the lower ones to the RITAP solution. (b) Dependence on $\rho$ of the ratio between the TAP optimal energy and the energy obtained after projecting the TAP solution onto an integer one using RITAP. (c)  Fraction of non-integer paths in TAP as a function of $\eta$. The dashed line has a slope of $1/\eta$ (equivalently $1/M$), to illustrate asymptotic behavior. (d) Average degeneracy of support paths as a function of $\eta$. In the case of Ahaheim the computation of the support paths was too lengthy computationally, therefore we instead computed the number of paths over which TAP flow is routed in the FW solution. This is found to agree well with the actual degeneracy in the cases (low $\eta$) in which it could be computed.}
    \label{fig:real_nets_highM}
\end{figure}
The quantities that are plotted are the same as in figure \ref{fig:energy_conv_tap_itap}. 
As in the RRG case the ITAP energy approaches the TAP one (panels (a,b)) when $\eta$ increases. Moreover the degeneracy $K-1$ plateaus at high $M$ (panel (d)) thus producing the $1/\eta$ rate in the quantity $1-F_\text{int}$ (panel (c)). 
Looking at the insets in panels (a,b,c) we remark that when $\rho=\Theta(1)$ the curves overlap when plotted as a function of $\rho$. Instead, considering the regime $\eta=\Theta(1)$ we see that curves especially in panels (b,c) collapse when plotted against $\eta$, with power law behavior arising for $\eta\approx 1$.
This confirms that $\eta$ and $\rho$ and are the relevant parameters when studying the asymptotic and fixed flow per edge regimes. To conclude, we also ran additional numerical experiments with typical nonlinearities borrowed from traffic assignment. We remark that the above phenomenology is quite robust with respect to the functional form of the nonlinearity. See Appendix \ref{app:num_exp_M=N^2} for the detailed results.

\section{Conclusions}
Our work provides an examination of the integer traffic assignment problem on sparse graphs, with origin-destination pairs picked uniformly at random. 
We first devised several algorithms, most of them based on the idea of optimizing the position of one path while keeping the others fixed. The only exception is RITAP which is based on relaxing the integral flow constraint. We found a more transparent derivation of CBP, and proposed a novel method to sample self-avoiding paths on arbitrary graphs. 

To compare the algorithms we focused on the regime where the flow per edge is of order one, $\rho\coloneqq 2M\log N/(Nd\log d)=O(1)$. 
In the attractive case ($\phi$ concave), we see that simulated annealing and CBP are the most effective. CBP, however, fails to converge in some settings.

In the repulsive case ($\phi$ convex) we find that all algorithms except for RITAP attain approximately the same performance, indicating that elaborated methods are not significantly more performant than simple ones like the greedy algorithm. Only simulated annealing exhibits a marginally better performance when the degree of the RRG increases. 

We conclude that in the absence of computational cost constraints, simulated annealing represents a valid choice to solve ITAP. When taking speed into consideration, we find that the greedy algorithm and RITAP are orders of magnitude faster while sacrificing only a few per cent in energy. RITAP in particular benefits from an implementation that makes its running time independent of $M$, when $M>N^2$. It is thus better suited to cases where $M$ is very large. Finally, we remark that CBP is also effective, but its slowness and unreliable convergence limit its applicability. 

Next, using RITAP, we investigate the regime $M=\Theta(N^2)$ where every entry of the demand matrix is of order one. We do so in the case $\phi(x)=x^2$. We observe that as $\eta=M/(N(N-1))$ grows, TAP approaches ITAP. Therefore, if we are interested in solving the NP-hard problem ITAP in this limit, we can instead solve TAP, which is a convex problem and the relaxed solution will be already very close to an integer one. 
The convergence of TAP to ITAP is well explained by the fact that the degeneracy of the support paths plateaus when $\eta > 1$. We conjecture that the degeneracy becoming constant for higher $M$ originates from the demand matrix becoming more uniform (in the sense that all entries have similar magnitude). In particular the distribution of OD pairs, rather than their number plays a crucial role in controlling the degeneracy. 

By fixing a graph and gradually increasing $M$ our study characterizes the phases encountered by the system. One first finds the regime $\rho=\Theta(1)$, where paths start interacting. Then, as $\eta$ becomes of order one, the system enters in an asymptotic regime where most relevant quantities either plateau (e.g. $1-H/H_\text{SP}$, $K$) or exhibit a power law behavior (e.g. $H_\text{ITAP}/H_\text{TAP}-1$, $1-F_\text{int}$). This phase is characterized by the aforementioned convergence of TAP to ITAP. The relevance and properties of the two scales, $\rho = O(1)$ and $\eta = O(1)$, were further verified in the case of real-world graphs. In future work it would be important to consider realistic and not random origin-destination pairs and investigate how that changes the properties of the problem. 

\section*{Acknowledgements}
We thank Bill Chi Ho Yeung for sharing with us the code of the CBP algorithm.
\bibliography{biblio}
\appendix
\section{Pseudocode of the greedy algorithm}
\label{app:pseudocode_greedy}
The following is a pseudocode of an efficient implementation of the greedy algorithm.
\begin{algorithm}[H]
\caption{Greedy algorithm} 
\label{alg:greedy_algorithm}
\begin{algorithmic}[1]
\State\textbf{Input: } $G=(V,E)$ graph over which paths are sampled, $\{s^\mu\}_{\mu\in[M]},\{t^\mu\}_{\mu\in [M]}$ respectively the origins and destinations of each path,
$\{\pi^\mu_0\}_{\mu\in [M]}$ set of initial paths.
\State\textbf{Output:} A set of paths $\{\pi^\mu_0\}_{\mu\in [M]}$, that correspond to a local minimum of $H$.
\State Gflow$\gets$ copy($G$) \Comment{the weight on edge $e$ of Gflow is $I_e$} 
\State Gcost $\gets$ copy($G$)  \Comment{the weight on edge $e$ of Gcost is $\phi(I_e)$}
\For{$e \in G.$edges}
\State Gflow.weight($e$)$\gets$0
\EndFor
\For{$\mu=1,\dots,M$}
\State $\pi^\mu\gets \pi_0^\mu$
\For{$k=1,\dots,$ length($\pi^\mu)-1$}
\State $e\gets [\pi^\nu[k],\pi^\nu[k+1]]$
\State Gflow.weight($e$)$\gets$ Gflow.weight($e$)+1 \Comment{populating the flow graph with paths}
\EndFor
\EndFor
\State H$\gets 0$ 
\For{$e \in G.$edges} \Comment{loop to compute current value of $H(\pmb I)$}
\State Gcost.weight($e$)$\gets\phi(\text{Gflow.weight($e$)}+1)-\phi(\text{Gflow.weight($e$)})$
\State H$\gets$ H+Gcost.weight($e$)
\EndFor
\State flagconv $\gets$ False\Comment{flag for convergence}
\While{not flagconv}\Comment{main loop}
\For{$\nu=1,\dots,M$}\Comment{loop over all paths}
\For{$k=1,\dots,\text{length($\pi^\nu$)}-1$}
\State $e\gets [\pi^\nu[k],\pi^\nu[k+1]]$
\State Gflow.weight($e$)$\gets$ Gflow.weight($e)-1$\Comment{removing path $\nu$ from the graph so that one has $I^{\backslash\nu}$}
\State Gcost.weight($e$)$\gets\phi(\text{Gflow.weight($e$)}+1)-\phi(\text{Gflow.weight($e$)})$
\EndFor
\State $\pi^\nu\gets$  shortestPath(Gcost, $s^\nu,t^\nu$)\Comment{updating path $\nu$}
\For{$k=1,\dots,\text{length($\pi^\nu$)}-1$}
\State $e\gets [\pi^\nu[k],\pi^\nu[k+1]]$
\State Gflow.weight($e$)$\gets$ Gflow.weight($e$)+1\Comment{adding back path $\nu$ to the graph}
\State Gcost.weight($e$)$\gets\phi(\text{Gflow.weight($e$)}+1)-\phi(\text{Gflow.weight($e$)})$
\EndFor
\EndFor
\State newH$\gets 0$ \Comment{measuring energy after all paths are updated}
\For{$e \in G.$edges}
\State newH$\gets$newH+Gcost.weight($e$)
\EndFor
\If{H == newH}
\State flagconv$\gets$True \Comment{If $H(\pmb I)$ does not change once all the paths are updated, then the algorithm has converged}
\EndIf
\State H $\gets$ newH
\EndWhile
\end{algorithmic}
\end{algorithm}
The method '.weight($e$)' accesses the weight of edge $e$. '$G$.edges' is a list of all the edges in $G$. Each path is represented as a list of nodes, i.e., $\pi=[\pi[1],\pi[2],\dots,\pi_{\text{length}(\pi)}]$. Similarly an edge is represented as a list of two nodes. The function 'shortestPath($G,s,t$)' instead returns she shortest path from $s$ to $t$ in the weighted graph $G$.

\section{Derivation of CBP}
\label{app:BP_derivation}

The conditional probability distribution \eqref{eq:marginal_directed} can be represented by a factor graph, where variable nodes are the at the edges and factor nodes at the nodes of the original network. Let $(ij)$ be an (undirected) edge on the graph $G$. We place two variables on this edge $\sigma_{i\to j}^\nu =\mathbb{I}[i\to j\in \pi^\nu]$ and $\sigma_{j\to i}^\nu =\mathbb{I}[j\to i\in \pi^\nu]$.
The probability distribution (\ref{eq:conditional_prob_with_edgecosts})can be decomposed into factors as $P\left(\{\sigma^{\nu}_{i\to j}\}_{(ij)\in E}|\{H^{\nu}_{ij}\}_{ (ij)\in E},s^{\nu},t^{\nu}\right) \propto f_{s^\nu} f_{t^\nu} \prod_{i\in [N]\backslash s^\nu,t^\nu} f_i $. 
Let us look at these factors:
\begin{itemize}
    \item $f_i=\mathbb{I}[\sum_{\substack{k,l\in \partial i\\k\neq l}}\sigma^{\nu}_{k\to i}+\sigma^{\nu}_{i\to l}=(0 \text{ OR } 2)] \exp{(-\beta \sum_{k\in \partial i} H^{\nu}_{ik} \sigma^{\nu}_{i\rightarrow k} )}]$.
    This constraint ensures that $\pi^\nu$ passes through node $i$ no more than once. It also ensures that the $\{\sigma_{i\to j}\}$ form a valid path and it contains the weighting factor $e^{-\beta H^{\nu}_{ik} }$ that represents the cost of passing through edge $i,k$, independently of direction. This constraint is placed at every node in the graph except for $s^\nu,t^\nu$
    \item $f_{t^\nu}=\mathbb{I}\left[\sum_{k\in\partial t^\nu}\sigma^{\nu}_{k\to     t^\nu}=1\right]\mathbb{I}\left[\sum_{k\in\partial t^\nu}\sigma^{\nu}_{t^\nu\to k}=0\right] $ imposes that the path ends at $t^\nu$.
    \item $f_{s^\nu}=\mathbb{I}\left[\sum_{k\in\partial s^\nu}\sigma^{\nu}_{s^\nu\to k}=1\right]\mathbb{I}\left[\sum_{k\in\partial s^\nu}\sigma^{\nu}_{k\to s^\nu}=0\right]\prod_{k \in \partial s^\nu}  ( \exp{(-\beta H^{\nu}_{s^\nu k}\sigma^{\nu}_{s^\nu \rightarrow k})} )$ imposes that the path starts at $s^\nu$.
\end{itemize}
Having defined the quantity $f_i$, one can draw the factor graph as shown in Figure \ref{fig:factor_graph}.

 \begin{figure} [H]
     \centering
     \includegraphics[width=0.38\linewidth]{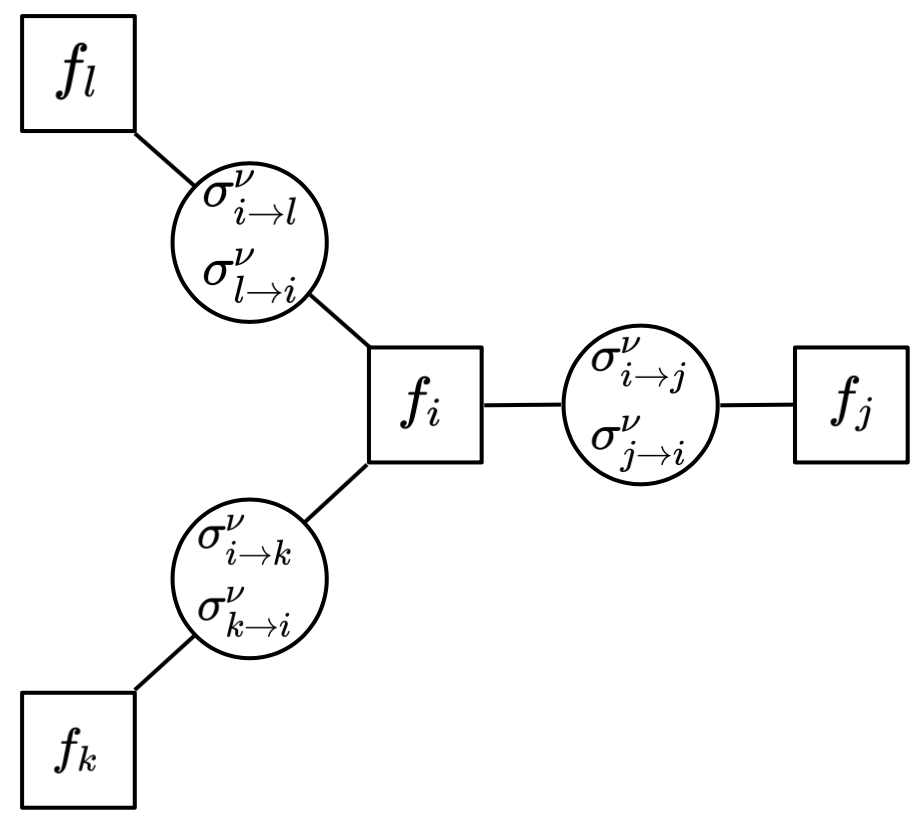}
     \caption{Factor graph associated to the probability conditional probability distribution (\ref{eq:conditional_prob_with_edgecosts}). }
     \label{fig:factor_graph}
\end{figure}

From the construction of the graph, one can deduce that if the original network is tree like, the factor graph will remain tree like given our choice of variables. Given this property, BP equations are expected to converge to a good solution if the graph does not have too many loops.

We can now derive the BP equations for this factor graph; this gives two sets of messages. Since each variable node has exactly two neighbouring factor nodes, simplifications of the BP equations leads to the iteration of only one set of messages. We denote by $\chi^{\nu}_{i\rightarrow j}$ the message going from node $i$ to node $j$. The message $\chi^{\nu}_{i\rightarrow j}(\sigma^{\nu}_{i\to j},\sigma^{\nu}_{j\to i})$ depends on the value of the variable, that can take four different values. First notice that $\nu_{i\to j}(1,1)=0$, since the path cannot backtrack on one edge. For the other three values ($\sigma^\nu_{i\to j},\sigma^\nu_{j\to i}) \in$ \{(0,0),(1,0),(0,1)\} (for a node $i\neq s^\nu,e^\nu$), update equations are derived: 

 \begin{align}
     \chi^{\nu}_{i\to j}(0,0)&=\frac{1}{Z_{i\to j}}\left[\prod_{k\in\partial i\backslash j}\chi^{\nu}_{k\to i}(0,0)+\sum_{\substack{k,l\in\partial i\backslash j\\ k\neq l}} e^{-\beta H^\nu_{il}}\chi^{\nu}_{k\to i}(1,0)\chi^{\nu}_{l\to i}(0,1)\prod_{r\in \partial i\backslash j,k,l}\chi^{\nu}_{r\to i}(0,0)\right] \nonumber  \\
     \chi^{\nu}_{i\to j}(1,0)&=\frac{e^{-\beta H^\nu_{ij}}}{Z_{i\to j}}\sum_{k\in\partial i\backslash j}  \chi^{\nu}_{k\to i}(1,0)\prod_{l\in \partial i\backslash j,k}\chi^{\nu}_{l\to i}(0,0)\\
     \chi^{\nu}_{i\to j}(0,1)&=\frac{1}{Z_{i\to j}}\sum_{k\in\partial i\backslash j} e^{-\beta H^\nu_{ik}} \chi^{\nu}_{k\to i}(0,1)\prod_{l\in \partial i\backslash j,k}\chi^{\nu}_{l\to i}(0,0) \nonumber
\label{BP_chi}
\end{align}

Equations \ref{BP_chi} are intuitive. For example, to derive the expression of $\chi^{\nu}_{i\to j}(0,0)$, two cases need to be considered: the case where none of the neighboring edges to $i$ is occupied, and the case where a path passes through node $i$, but without passing through node $j$ (and therefore passes through $k\to i\to l$ with $k,l\neq j$). Same reasoning can be applied for the two other messages.

Special equations apply for the source node $s^{\nu}$: 

\begin{align}
    \chi^{\nu}_{s^{\nu}\to j}(0,0)&=\frac{1}{Z_{s^\nu\to j}} \sum_{k\in\partial s^{\nu}\backslash j} e^{-\beta H^\nu_{s^{\nu}k}}\chi^{\nu}_{k\to s^{\nu}}(0,1) \prod_{l\in \partial s^{\nu}\backslash j,k}\chi^{\nu}_{l\to s^{\nu}}(0,0) \nonumber \\
    \chi^{\nu}_{s^{\nu}\to j}(1,0)&=e^{-\beta H^\nu_{s^{\nu}j}} \frac{1}{Z_{s^\nu\to j}} \prod_{k\in \partial s^{\nu}\backslash j}\chi^{\nu}_{k\to s^{\nu}}(0,0)\\
    \chi^{\nu}_{s^{\nu}\to j}(0,1)&=0 \nonumber 
\end{align}

and for the end node $e^{\nu}$ :
\begin{align}
    \chi^{\nu}_{e^{\nu}\to j}(0,0)&=\frac{1}{Z_{e^{\nu}\to j}} \sum_{k\in\partial e\backslash j}\chi^{\nu}_{k\to e^{\nu}}(1,0) \prod_{l\in \partial e\backslash j,k}\chi^{\nu}_{l\to e^{\nu}}(0,0) \nonumber \\
    \chi^{\nu}_{e^{\nu}\to j}(0,1)&=\frac{1}{Z_{e^\nu\to j}} \prod_{k\in \partial e\backslash j}\chi^{\nu}_{k\to e^{\nu}}(0,0) \\
    \chi^{\nu}_{e^{\nu}\to j}(1,0)&=0 \nonumber
\end{align}

To get rid of the normalization factor $1/Z_{i\rightarrow j}$ present in the messages, the following change of variables can be applied:
\begin{align}
\label{eq:changeofvariable}
    z^{\nu}_{i\rightarrow j } = \frac{\chi^{\nu}_{i\rightarrow j}(1,0)}{\chi^{\nu}_{i\rightarrow j}(0,0)} \text{ and } \tilde{z}^{\nu}_{i\rightarrow j } = \frac{\chi^{\nu}_{i\rightarrow j}(0,1)}{\chi^{\nu}_{i\rightarrow j}(0,0)} e^{-\beta H^\nu_{ji}}
\end{align}

Using equations \ref{BP_chi}, one finds a closed form for the normalized messages $z^{\nu}_{i\rightarrow j }$ and $\tilde{z}^{\nu}_{i\rightarrow j }$ :

\begin{align}
\label{eq:zupdate}
    z^{\nu}_{i\to j}&= \I(i\neq s^{\nu})\I(i\neq e^{\nu})\frac{e^{-\beta H^\nu_{ij}}\sum_{k\in\partial i\backslash j} z^{\nu}_{k\to i}}{1+\sum_{\substack{k,l\in\partial i\backslash j\\k\neq l}} z^{\nu}_{k\to i}\tilde z^{\nu}_{l\to i}} + \I(i= s^{\nu}) \frac{e^{-\beta H^\nu_{s^{\nu}j}}}{\sum_{k\in\partial s^{\nu}\backslash j}  \tilde z^{\nu}_{k\to s^{\nu}}} \\ 
\label{eq:ztildeupdate}
    \tilde{z}^{\nu}_{i\to j}&= \I(i\neq s^{\nu})\I(i\neq e^{\nu}) \frac{\sum_{k\in\partial i\backslash j} \tilde{z}^{\nu}_{k\to i}}{1+\sum_{\substack{k,l\in\partial i\backslash j\\k\neq l}} z^{\nu}_{k\to i}\tilde z^{\nu}_{l\to i}} + \I(i= e^{\nu}) \frac{1}{\sum_{k\in\partial e^{\nu}\backslash j} z^{\nu}_{k\to e^{\nu}}}
\end{align}

Equations \ref{eq:zupdate} and \ref{eq:ztildeupdate} are fixed point equations for the messages $z^{\nu}_{i\to j}$ and $\tilde{z}^{\nu}_{i\to j}$ at finite $\beta$. In order to take the $\beta \rightarrow \infty$ limit, it is useful to define the quantities $a^{\nu}_{i\to j}=-\log(z^{\nu}_{i\to j})/\beta$ and $b^{\nu}_{i\to j}=-\log(\tilde z^{\nu}_{i\to j})/\beta$. After taking the zero temperature limit ($\beta\to\infty$), the belief propagation equations thus provide update equations for the messages $a^{\nu}_{i\rightarrow j}$ \ref{eq:BP_a} and $b^{\nu}_{i\rightarrow j}$ \ref{eq:BP_b}. 

The zero temperature belief propagation equations presented allow, if iterated until convergence, to find the path $\nu$ minimizing the energy once all other paths $\mu\in [M]\backslash \nu$ are fixed. Indeed, it samples from the conditional probability distribution \ref{eq:conditional_prob_with_edgecosts} that concentrates around the configurations of lowest energy when $\beta$ goes to $\infty$. As already mentioned in section \ref{sec:CBP}, following this procedure results in one step of the greedy procedure. In other words, BP is used as an algorithm to find the shortest path. 
However, the procedure described in \cite{yeung2013physics} consists in only making one iteration of the update equations \ref{eq:BP_a} and \ref{eq:BP_b} before updating the cost on edges. The cost $H_{ij}^\nu$ at the edge $(ij)$ being defined as $H_{ij}^\nu = \phi(I^{\backslash \nu}_{ij}+1)-\phi(I^{\backslash \nu}_{ij})$, it is necessary to compute the expected total traffic of the other paths $I^{\backslash \nu}_{ij}$ between $i$ and $j$. We can write the expected traffic as $I^{\backslash \nu}_{ij} = \sum_{\mu\in [M]\backslash \nu} P(\sigma^\mu_{(ij)}=1)$ where $\sigma^\mu_{(ij)}=\I[ i\to j \in \pi^\nu \text{ OR }  j\to i \in \pi^\nu]$. The expected contribution of path $\mu$ to the traffic at edge $(ij)$ is equal to the probability $P(\sigma^{\mu}_{(ij)}=1)$. Naturally, we have that $P(\sigma^{\mu}_{(ij)}=1)=P(\sigma^{\mu}_{i\rightarrow j}=1) + P(\sigma^{\mu}_{j\rightarrow i}=1)$, where $P(\sigma^{\mu}_{i\rightarrow j}=1)$ can be expressed using the messages $a^{\mu}_{i\rightarrow j}$ and $b^{\mu}_{i\rightarrow j}$. In order to derive the expression, one can start by deducing the marginal from the finite temperature messages $\chi^{\mu}_{i \rightarrow j}$, then operate the same change of variables \ref{eq:changeofvariable} as follows : 
\begin{align}
P(\sigma^\mu_{i\rightarrow j}=1)& =\frac{{\chi^{\mu}_{i\rightarrow j}(1,0)} \sum_{k\in \partial j\backslash i} e^{-\beta H^\mu_{jk}} \chi^{\mu}_{k\rightarrow j}(0,1) \prod_{l\in \partial j\backslash i,k} \chi^{\mu}_{l\rightarrow j }(0,0)}{ \prod_{k\in \partial j} \chi^{\mu}_{k\rightarrow j}(0,0) + \sum_{k,l\in\partial j, k\neq l} e^{-\beta H^\mu_{jl}} \chi^{\mu}_{k\rightarrow j}(1,0) \chi^{\mu}_{l \rightarrow j}(0,1) \prod_{r\in \partial j \backslash l,k} \chi^{\mu}_{r\rightarrow j } (0,0)} \nonumber  \\
& = \frac{z^\mu_{i\rightarrow j } \sum_{k\in \partial j \backslash i} \tilde{z}^{\mu}_{k\rightarrow j}}{ 1 + \sum_{k,l\in \partial j, k\neq l} z^{\mu}_{k\rightarrow j } \tilde{z}^{\mu}_{l\rightarrow j} }
\end{align}

Taking the limit $\beta\rightarrow \infty$, the probability $P(\sigma^\mu_{i\rightarrow j}=1)$ concentrates at one value (either 0 or 1) depending on the values of the messages. Explicitly, the value of $\sigma^\mu_{i\rightarrow j}$ (including the case $i=s^\mu$ and $i=t^\mu$) will follow the expression: 
\begin{align}
  \sigma^\mu_{i\rightarrow j} = \I(i=s^{\mu})&\Theta\bigg(
  \min_{l\in\partial j\backslash i} (b^\mu_{l\rightarrow j})-b^\mu_{i \rightarrow j}
  \bigg) \nonumber \\ & + \I(i\neq s^{\mu})\Theta\bigg( \min{\bigg( 0,\min_{\substack{l\in \partial j \backslash i}}{(a^\mu_{l\rightarrow j } + \min_{\substack{k \in \partial j \backslash l} } b^\mu_{k\rightarrow j})} \bigg)} -(a^\mu_{i \rightarrow j}+ \min_{\substack{k \in \partial j\backslash i}} b^\mu_{k\rightarrow j}) \bigg) 
\label{marginal_directed}
\end{align}
The results derived in this appendix can be compared to the ones exposed in the supplementary material of \cite{yeung2013physics} and to the pseudocode of the CBP algorithm (also taken from \cite{yeung2013physics}). The update equations for $a^\nu_{i\rightarrow j}$ \eqref{eq:BP_a} and $b^\nu_{i\rightarrow j}$ \eqref{eq:BP_b} can be compared to equations \eqref{eq3CBP},\eqref{eq4CBP} which are the ones originally derived in \cite{yeung2013physics}.
The two sets of equations, are almost identical, we shall now focus on the differences. The main difference is the presence of $H_{ij}^\nu$ in place of $\phi'(\eta_{ji}^{\nu*})$.
This is due to the fact that the authors use an approximated form of the energy obtained by Taylor expanding $H_{ij}^\nu$. 
This approximation is only valid if $I_{i\rightarrow j}^{\backslash \nu}\gg 1$ which is not the case in the regime $\rho=O(1)$, which is the one considered in \cite{yeung2013physics}. Another difference lies in the update equation of message $b^\nu_{i\rightarrow j}$ in the case $i=e^\nu$. Indeed, the term $-\min_{k\in \partial e^\nu \backslash j } [a^\nu_{k\rightarrow e^\nu}]$ is present in equation \eqref{eq:BP_b} and is missing from equation \eqref{eq4CBP}.
Additionally, some algebraic computations reveal that equation \eqref{eq:marginal_directed} exactly corresponds to equation \eqref{eq1CBP}, originally derived in \cite{yeung2013physics}.

Despite the presence of these differences, in our experiments we use the code from \cite{yeung2013physics}. We further tried using the exact formula for  $H_{ij}^\nu$ in place of the Taylor expansion and this did not improve the energy reached by CBP.
\subsection{CBP pseudocode}
\label{sec:pseudocode_CBP}
The pseudo-code corresponding to the implementation made by the authors of \cite{yeung2013physics} is given in \ref{alg:CBP_pseudocode}. 
We report below the formulas to update the binary variables of each path based on the values of the messages. We say $\sigma_{ji}^\nu$ is a binary variable indicating if path $\nu$ goes through the undirected edge $(ji)$ and $\eta_{ji}^{\nu *}$ the expected traffic through edge $(ji)$ taking into account that path $\nu$ passes through it (thus the +1). We can determine these quantities from the messages as
\begin{align}
\label{eq1CBP}
 \sigma_{j i}^\nu & = \I(i=s^{\nu}) \Theta\left(\min _{l \in \partial j \backslash i}\left[b_{l \rightarrow j}^\nu\right]-b_{i \rightarrow j}^\nu\right)+\I(i=t^{\nu}) \Theta\left(\min _{l \in \partial j \backslash i}\left[a_{l \rightarrow j}^\nu\right]-a_{i \rightarrow j}^\nu\right) \nonumber \\
& + \I(i\neq s^{\nu})\I(i\neq t^{\nu}) \Theta\left(\min \left[0, \min _{\substack{l, r \in \partial j \backslash i \\
l \neq r}}\left[a_{l \rightarrow j}^\nu+b_{r \rightarrow j}^\nu\right]\right]\right. \nonumber \\
&\left.-\min \left[a_{i \rightarrow j}^\nu+\min _{l \in \partial j \backslash i}\left[b_{l \rightarrow j}^\nu\right], b_{i \rightarrow j}^\nu+\min _{l \in \partial j \backslash i}\left[a_{l \rightarrow j}^\nu\right]\right]\right),
\end{align}
\begin{equation}
\label{eq2CBP}
\eta_{j i}^{\nu *}=1+\sum_{\mu \neq \nu} \sigma_{j i}^\mu,
\end{equation}
The last two equations are respectively referenced to in the pseudo-code by formulas $F_1(\{a_{r\rightarrow j}^\nu\}_{r\in \partial j},\{b_{s\rightarrow j}^\nu\}_{s\in \partial j})$ and $F_2(\{\sigma^\mu_{(ji)}\}_{\mu\in [M]\backslash \nu})$.

We now give the update equations for the sets of messages $a$ and $b$.
\begin{align}
\label{eq3CBP}
a_{j \rightarrow i}^\nu= \begin{cases}\min _{l \in \partial j \backslash i}\left[a_{l \rightarrow j}^\nu\right]+\phi^{\prime}\left(\eta_{j i}^{\nu *}\right)-\min \left[0, \min _{\substack{l, r \in \partial j \backslash i \\
l \neq r}}\left[a_{l \rightarrow j}^\nu+b_{r \rightarrow j}^\nu\right]\right], & \text{if } j\not \in \{s^{\nu},e^{\nu}\} \\
-\min_{l \in \partial j \backslash i}\left[b_{l \rightarrow j}^\nu\right]+\phi^{\prime}\left(\eta_{j i}^{\nu *}\right), & \text{if } j =s^{\nu}   \\
\infty, & \text{if } j = e^{\nu} \end{cases} \\
b_{j \rightarrow i}^\nu= \begin{cases}\min _{l \in \partial j \backslash i} \left[b_{l \rightarrow j}^\nu\right]+\phi^{\prime}\left(\eta_{j i}^{\nu *}\right)-\min \left[0, \min_{\substack{l, r \in \partial j \backslash i \\
l \neq r}}\left[a_{l \rightarrow j}^\nu+b_{r \rightarrow j}^\nu\right]\right], & \text{if } j\not \in \{s^{\nu},e^{\nu}\} \\
\infty, & \text{if } j =s^{\nu} \\
\phi^{\prime}\left(\eta_{j i}^{\nu *}\right), & \text{if } j =e^{\nu} 
\end{cases}
\label{eq4CBP}
\end{align} 
where the term $\phi^{\prime}\left(\eta_{j i}^{\nu *}\right)$ quantifies the contribution for the total energy of path $\nu$ passing through the undirected edge $(ji)$, using the Taylor expansion as explained previously. 
Formula $F_3(\{a_{r\rightarrow i}^\nu\}_{r\in \partial i\backslash j},\{b_{s\rightarrow i}^\nu\}_{s\in \partial i\backslash j},\phi'(\eta_{ji}^{\nu *}))$ allows to update the messages and is reported in the equations \eqref{eq3CBP} and \eqref{eq4CBP}. 
The initialization of the messages is chosen in order to bias the algorithm to converge towards a configurations with consistent paths and this is done in the following way. Let us consider the path $\nu$. While the messages related to path $\nu$ involving all nodes except $s^\nu$ and $e^\nu$ are initialized to a positive random value between 0 and 1. Instead the messages $\{a_{s^\nu \rightarrow j}^\nu\}_{j\in \partial s^\nu}$ exiting source node are set to a random negative value between -1 and 0 (the probability for the edge to be occupied is higher). Also, given that path $\nu$ is not allowed to exit the end node $t^\nu$, it is natural to initialize the messages $\{a_{t^\nu \rightarrow j}^\nu\}_{j\in\partial t^\nu}$ to infinity (corresponding to a probability 0). In the same way, all messages in $\{a_{s^\nu \rightarrow j}^\nu\}_{j\in \partial s^\nu}$ are initialized to infinity as no value as edges entering node $s^\nu$ should not be occupied by the path $\nu$. Some other implementation choices have been made in order to improve the convergence of the algorithm. For example, an additive noise has been introduced in the computation on the costs on the edges (term $\phi'(\eta_{ij}^{\nu *})$ line 13 of pseudo-code \ref{alg:CBP_pseudocode}) in order to break the degeneracy (cases where two path configurations have the same energy). 

\begin{algorithm}[H]
\caption{CBP algorithm for routing}
\label{alg:CBP_pseudocode}
\begin{algorithmic}[1]
\State\textbf{Input: } $G=(V,E)$ graph over which paths are computed. Sources $\{s^\mu\}_{\mu\in [M]}$ and destinations $\{s^\mu\}_{\mu\in [M]}$ of each path. The maximum number of iterations $MaxIter$. Equilibrium time $EqTime$ (stopping condition). A tolerance level $\epsilon$ to determine if the messages have converged. 

\State\textbf{Output:} $\{\sigma_{(ij)}^\nu\}$ are binary variables representing the occupancy of each edge by each path (=1 if undirected edge $(ij)$ is occupied by path $\nu$ and 0 otherwise). 
\State For all $(ij)\in E,\nu \in [M]$ initialize messages $a^\nu_{i\rightarrow j}$ and $b^\nu_{i\rightarrow j}$ as : 
\begin{align*}
&\begin{aligned}
a^\nu_{i\to j}&=\left\{
\begin{array}{lr}
    - U[0,1]\text{ \ \ if $i= s^\nu$ }, \\
    \infty  \text{ \ \ \ \ \ \ \ \  \ \ if $i= t^\nu$ },\\ 
    U[0,1] \text{ \ \  \ \ \ otherwise.}
\end{array}\right.
\end{aligned}
&
&\begin{aligned}
b^\nu_{i\to j}&=\left\{
\begin{array}{lr}
    \infty  \text{ \ \ \ \ \ \ \ \ if $i= s^\nu$ }, \\ 
    0   \text{\ \ \ \ \ \ \ \ \ \ \ if $i= t^\nu$ },\\ 
    U[0,1] \text{ \ \ \ otherwise.}
\end{array}\right. \nonumber
\end{aligned}
\end{align*} 

\State $t$ $\gets $ 0  \Comment{counter on number of iterations.}
\State $EqCount$ $\gets$ 0  \Comment{counter used to condition on the convergence. } 
\Repeat
    \State Select a random node $i \in V$ and a random neighbour $j \in \partial i $
    \State $pre_a^\nu$ $\gets $ $a_{i\rightarrow j}^\nu \text{\ \ , \ \ } \forall \nu\in [M] $  
    \State $pre_b^\nu$ $\gets $ $b_{i\rightarrow j}^\nu \text{\ \ , \ \ } \forall \nu\in [M] $ \Comment{Store old messages.}
    \For{each path $\nu \in [M]$} 
        \State $\sigma_{(ij)}^\nu  \gets$ $F_1(pre_a^\nu,pre_b^\nu,\{a_{r\rightarrow j}^\nu\}_{r\in \partial i\backslash j},\{b_{s\rightarrow j}^\nu\}_{s\in \partial i\backslash j})$ \Comment{assign spins, based on old messages.} 
        \State $\eta_{ji}^{\nu *} \gets$ $F_2(\{\sigma^\mu_{(ji)}\}_{\mu\in [M]\backslash \nu})$ \Comment{Compute the undirected traffic at edge $(ji)$ ignoring path $\nu$.}
        \State $a_{i\rightarrow j}^\nu, b_{i\rightarrow j}^\nu \gets  $ $F_3(\{a_{r\rightarrow i}^\nu\}_{r\in \partial i\backslash j},\{b_{s\rightarrow i}^\nu\}_{s\in \partial i\backslash j},\phi'(\eta_{ji}^{\nu *}))$ \Comment{Update of the messages.}
    \EndFor
    \If {$\sum_{\mu=0}^N |a_{j\rightarrow i}^\mu-pre_a^\mu|+|b_{j\rightarrow i}^\mu-pre_b^\mu| < \epsilon$} \Comment{Check for convergence} 
        \State $EqCount \gets EqCount+1$
    \Else 
        \State $EqCount \gets 0$ 
    \EndIf  
    \State $t \gets t+1$
\Until{($EqCount$ $\geq$ $EqTime$ or $t \geq$ $MaxIter*$card($V$))}
\For {each edge $(ij)\in E$}
    \For {each path $\nu \in [M]$}
        \State $\sigma_{(ij)}^\nu  \gets$ $F_1(\{a_{r\rightarrow j}^\nu\}_{r\in \partial j},\{b_{q\rightarrow j}^\nu\}_{q\in \partial j})$
    \EndFor
\EndFor
\end{algorithmic}
\end{algorithm}
\section{Derivation of the SAW sampling algorithm}
\label{app:SAW_sampler}
This appendix we detail the derivation of the MCMC to sample SAWs on a weighted graph.
Let us first set up some notation. We will indicate with $\pi=(\pi_1,\dots,\pi_{L(\pi)})$ a path going from $\pi_1$ to $\pi_{L(\pi)}$. $L(\pi)$ indicates the number of nodes in the path. Consider a graph with edges weighted with the matrix $W\in\R^{N\times N}$. $W_{ij}$ represents the cost for the walk to traverse edge $i\to j$. If the edge $i\to j$ is not present we set $W_{ij}=\infty$. 
We suppose $W_{ij}\geq0$ for all $i,j$. When referring to shortest paths, we will mean paths with minimum weight with respect to $W$. We wish to sample from the following probability measure, over the SAWs going from node $s$ to node $t$.
\begin{equation}\label{eq:path_prob_fixed_W}
P(\pi|s,t)=\frac{1}{Z(s,t)}e^{-\beta E_W(\pi)}\I[\pi_1=s]\I[\pi_{L(\pi)}=e]\I[\pi \text{ is a SAW}],\quad E_W(\pi)\coloneqq\sum_{k=1}^{L(\pi)-1} W_{\pi_k,\pi_{k+1}}
\end{equation}
When $\beta$ is large, the probability is concentrated around the shortest path (according to $E$). Instead when $\beta\to0$ the measure is uniform over all SAWs. To sample from this measure, we design a Metropolis based sampler. Let $\pi$ be the current path in the MCMC. We first propose a new path $\tilde \pi$ by sampling it from a proposal probability $Q(\tilde \pi|s,t)$ and then we accept the move with probability $p_{acc}=\frac{P(\tilde\pi|s,t)Q(\pi|s,t)}{P(\pi|s,t)Q(\tilde \pi|s,t)}$. Notice that differently from the usual Metropolis algorithm, the proposed path is independent from the current path.
With use the notation $\pi_{:k}$ and $\pi_{k:}$ we indicate respectively the first $k$ nodes of path $\pi$, and the nodes from $\pi_k$ (included) to the end of $\pi$. As a proposal probability we use an autoregressive process of the following form
\begin{equation}
\label{eq:proposal_saw_joint}
    Q(\tilde\pi|s,t)=\prod_{k=1}^{L(\pi)-1}q(\tilde\pi_{k+1}|\tilde\pi_{:k},s,t).
\end{equation} Thanks to the autoregressive property, $Q$ is easy to sample: it is sufficient to draw iteratively from the distribution $q(\cdot|\cdot)$ until the node $e$ is reached. Given a graph $G=(V,E)$ and a path (a sequence of nodes connected by edges) $\pi$ in $G$, we denote by $G\backslash\pi$ the graph where we have removed the nodes in $\pi$ and the edges connected to them. 
Define also $V_e(y,\pi)$ to be the cost of the shortest path from $y$ to $e$ in the graph $G\backslash\pi$. 
If such path does not exist, we define $V_e(y,\pi)$ to be $\infty$. we are now ready to define $q(\cdot|\cdot)$.
\begin{equation}
\label{eq:proposal_saw_conditional}
    q(\tilde\pi_{k+1}=y|\tilde\pi_{:k},s,t)=\frac{1}{\mathcal Z(x_{:k},s,t)}\I[y\not \in \tilde\pi_{:k}]\exp\left[-\beta( W_{\tilde\pi_k y}+V_e(y,\tilde\pi_{:k}))\right].
\end{equation}
In words, supposing we have sampled the first $k$ elements of the path and we must pick the $k+1$, for each neighbor $y$ of $\tilde\pi_k$, $W_{\tilde\pi_k y}+V_e(y,\tilde\pi_{:k})$ represents the cost of going from $x_k$ to $y$ and from there continuing along the shortest SAW from $y$ to $e$. As explained in section \ref{sec:high_beta_expansion}, this can be seen as a high $\beta$ expansion where one only keeps track of the first order term in $\beta$.  The pseudocode of this algorithm is presented in \ref{alg:SAW_sampler}.

\subsection{SAW sampler as high $\beta$ expansion}
\label{sec:high_beta_expansion}
In this section we show how the autoregressive process to propose SAWs described by \eqref{eq:proposal_saw_conditional} corresponds to a high $\beta$ expansion.
Let $P$ be the probability measure on SAWs with constrained ends defined in \eqref{eq:path_prob_fixed_W}.
Consider the following way of writing the conditional probability for $k+1$th node in the path given the first $k$ nodes, under $P$.
\begin{align}
\label{eq:high_beta_exp}
&P(\pi_{k+1}=y|\pi_{:k},s,t)=\\&=\sum_{\pi_{k+2},\dots,\pi_{L(\pi)}} P(\pi_{k+1:}=[y,\pi_{k+2:}]|\pi_{:k},s,t)=\frac{1}{Z(\pi_{:k},s,t)}\sum_{\pi_{k+2},\dots,\pi_{L(\pi)}}\I[\pi \text{ is a SAW}] \,e^{-\beta(W_{\pi_k y}+E_W(\pi_{k+2:}))}\\&=\frac{1}{Z(\pi_{:k},s,t)} e^{-\beta (W_{\pi_k y}+V_e(y,\pi_{:k}))}\sum_{\pi_{k+2},\dots,\pi_{L(\pi)}\in V}\I[\pi \text{ is a SAW}] \, e^{-\beta (E_W(\pi_{k+2:})-V_e(y,\pi_{:k}))}.
\end{align}
$Z(\pi_{:k,s,t})$ is as always a normalizing constant. In the second expression we sum over all the possible continuations of $\pi_{:k+1}$ and $[y,\pi_{k+2:}]$ indicates the path obtained by concatenating $y$ and $\pi_{k+2:}$, similarly in the following steps $\pi$ indicates the concatenation of $\pi_{:k},y,\pi_{k+2:}$. In the last passage we have brought out of the sum the term $V_e(y,\pi_{:k}))$ corresponding to the shortest path. 
Taking the limit $\beta\to\infty$ we obtain
\begin{equation}
    \lim_{\beta\to\infty}\sum_{\pi_{k+2},\dots,\pi_{L(\pi)}\in V}\I[\pi \text{ is a SAW}] \, e^{-\beta (E_W(\pi_{k+2:})-V_e(y,\pi_{:k}))}=\sum_{\pi_{k+2},\dots,\pi_{L(\pi)}\in V}\I[\pi \text{ is a SAW}]\I[E_W(\pi_{k+2:})=V_e(y,\pi_{:k}))]
\end{equation}
This is the number of shortest paths from $y$ to $e$ in the graph $G\backslash \pi_{:k}$. If one assumes $W$ to be a random matrix with continuous entries, then the shortest path is almost surely not degenerate and thus limit is one. Then it is clear that in the limit $\beta\to \infty$, $\frac{1}{\beta}\log P(\pi_{k+1}=y|\pi_{:k},s,t)\to \frac{1}{\beta}\log q(\pi_{k+1}=y|\pi_{:k},s,t)$, and therefore $q(\cdot|\cdot)$ captures the leading order in $\beta$. In some special cases cases there can be a huge degeneracy in the number of shortest paths: consider for example the case of $W$ being the adjacency matrix of a grid.

\begin{algorithm}[H]
\caption{Metropolis SAW sampler for weighted graph} 
\label{alg:SAW_sampler}
\begin{algorithmic}[1]
\State\textbf{Input: } Weight matrix $W\in\mathbb R^{N\times N}_+$, initial path $\pi_0$, inverse temperature $\beta$, $G$ graph over which the path is  sampled, $s,t$ respectively the origin and destination of the path, $t_\text{max}$ number of MCMC steps.
\State\textbf{Output:} a list $S$ of paths sampled from the Gibbs distribution \eqref{eq:path_prob_fixed_W}
\State $S=[\pi_0]$  
\State $\pi \gets \pi_0$
\State $K\gets \log Q(\pi|s,t)$ 
\For{$k=1,\dots,t_\text{max}$}
    \State $\tilde \pi\sim Q(\cdot|s,t)$ \Comment{Sampling proposed path using \ref{alg:proposal_sampler}. $Q$ depends implicitly on $W,G,\beta$.}
    \State $\Tilde K\gets \log Q(\tilde\pi|s,t)$
    \State $p_\text{acc}\gets \exp\left[K-\tilde K-\beta(E_W(\tilde \pi)-E_W(\pi))\right]$
    \State $U\sim\text{Uniform}([0,1])$
    \If{$U<p_\text{acc}$} \Comment{Accept-reject step}
    \State $\pi \gets \tilde\pi$
    \State $K\gets \tilde K$
    \EndIf
    \State $S.\text{append}(\pi)$
\EndFor
\end{algorithmic}
\end{algorithm}

\begin{algorithm}[H]
\caption{Proposal algorithm for SAW sampler} 
\label{alg:proposal_sampler}
\begin{algorithmic}[1]
\State\textbf{Input: } Weight matrix $W\in\mathbb R^{N\times N}_+$, inverse temperature $\beta$, $G$  directed weighted graph over which paths are sampled, $(s,t)$ respectively the origin and destination of the path.
\State\textbf{Output:} a SAW $\pi$ sampled from the proposal distribution \eqref{eq:proposal_saw_joint}
\State $\pi\gets[s]$  
\State $x\gets s$ \Comment{$x$ is the current endpoint of the path}
\While{$x\neq t$} \Comment{loop until the end node is reached}
    \For{$z\in\partial_{in}x$}\Comment{$\partial_{in} x$ is the set of nodes $z$ such that the edge $z\to x$ exists in $G$}
    \State $G$.removeEdge($z,x$) \Comment{Removing all incoming edges into $x$. This guarantees that the path is self avoiding}
    \EndFor 
    \State probnextnode $\gets[\;]$
    \State nextnodes $\gets[\;]$ \Comment{list of possible next nodes}
    \For{$y\in\partial_{out} x$ } \Comment{$\partial_{out} x$ is the set of nodes $z$ such that the edge $x\to z$ exists in $G$}
    \State nextnodes.append($y$)
    \State $V\gets$ shortestPath($G,W,y,t$) 
    \State probnextnode.append($e^{-\beta(V+W_{xy})}$)
    \EndFor
    \State probnextnode $\gets$ probnextnode/sum(probnextnode) \Comment{normalizing the probability}
    \State $i\sim$ probnextnode \Comment{sampling a index according to probnextnode}
    \State $x\gets$ nextnodes[i]
    \State $\pi.$append(x)
\EndWhile
\end{algorithmic}
\end{algorithm}
The function shortestPath($G,W,y,t$) returns $\min_{\pi: \pi_1=y, \pi_{L(\pi)=t}}E_W(\pi)$, where the minimum is taken over all the paths in graph $G$. Numerically the normalization step is tricky, since all the exponentials can be very negative, this might result in a $0/0$ division. To avoid this it is better to compute probabilities as using the formula $e^{-x_i}/\sum_j e^{-x_j}=e^{- (x_i-x_1)}/(1+\sum_{j=2}e^{-(x_j-x_1)})$, where we supposed $x_1=\min_i x_i$. In this way the expression is regularized. Notice that with a small variation, algorithm \ref{alg:proposal_sampler} can return also $Q(\pi|s,t)$, which is needed in the Metropolis step. One initializes $\log Q=0$ and every time a node is added one takes $\log Q=\log Q+\log($probnextnode[i]$)$, where is the index sampled from probnextnode the index of the next node in the path. Notice that in \ref{alg:proposal_sampler} we always assume the graph to be directed. If we are dealing with an undirected graph, then we just transform it into a directed one by replacing each undirected edge with two opposite directed ones.

\section{Derivation of the Simulated annealing algorithm}
\label{app:sim_annealing}
In this appendix we give the pseudocode of the simulated annealing algorithm used to find a solution to ITAP. As annealing schedule we use 
\begin{equation}
\label{eq:annealing_schedule}
    \beta(t)=\begin{cases*}
      \beta_\text{min}\frac{T_\text{ann}}{T_\text{ann}-t} & if $t<T_\text{ann}$ \\
      \infty      & if $t\geq T_\text{ann}$
    \end{cases*}
\end{equation}
In words, we start from $\beta_\text{min}$, then increase $\beta$, until it diverges after $T_\text{ann}$ steps, after that we keep $\beta=\infty$ until the iterations converge. 

\begin{algorithm}[H]
\caption{Simulated annealing for ITAP} 
\label{alg:simulated_annealing}
\begin{algorithmic}[1]
\State\textbf{Input: } Initial paths $\{\pi_0^\mu\}_{\mu\in[M]}$, OD pairs $\{s^\mu,t^\mu\}_{\mu\in[M]}$ inverse temperature schedule $\beta(t)$, $G$ graph over which paths are sampled, $t_\text{singlepath}$ number of Metropolis steps to make when sampling from \eqref{eq:conditional_gibbs}
\State\textbf{Output:} $\{\pi^\mu\}_{\mu\in[M]}$ a set of paths to which the algorithm converges
\State Gflow$\gets$ copy($G$) \Comment{the weight on edge $e$ of Gflow is $I_e$}
\State Gcost $\gets$ copy($G$)  \Comment{the weight on edge $e$ of Gcost is $\phi(I_e)$}
\For{$e \in G.$edges}
\State Gflow.weight($e$)$\gets$0
\EndFor
\For{$\mu=1,\dots,M$}
\State $\pi^\mu\gets \pi_0^\mu$
\For{$k=1,\dots,$ length($\pi^\mu)-1$}
\State $e\gets [\pi^\nu[k],\pi^\nu[k+1]]$
\State Gflow.weight($e$)$\gets$ Gflow.weight($e$)+1 \Comment{populating the flow graph with paths}
\EndFor
\EndFor

\State H$\gets 0$ \Comment{current value of $H(\pmb I)$}
\For{$e \in G.$edges}
\State Gcost.weight($e$)$\gets\phi(\text{Gflow.weight($e$)}+1)-\phi(\text{Gflow.weight($e$)})$
\State H$\gets$ H+Gcost.weight($e$)
\EndFor
\State flagconv $\gets$ False
\While{not flagconv}\Comment{main loop}
\For{$\nu=1,\dots,M$}\Comment{loop over all paths}
\For{$k=1,\dots,\text{length($\pi^\nu$)}-1$}
\State $e\gets [\pi^\nu[k],\pi^\nu[k+1]]$
\State Gflow.weight($e$)$\gets$ Gflow.weight($e$)-1\Comment{removing path $\nu$ from the graph so that one has $I^{\backslash\nu}$}
\State Gcost.weight($e$)$\gets\phi(\text{Gflow.weight($e$)}+1)-\phi(\text{Gflow.weight($e$)})$
\EndFor
\State $W\gets$ weightMatrix(Gcost)
\State$\pi^\nu\gets$ SAWsampler($W,\pi^\nu, \beta(t), G, s^\nu,t^\nu, t_\text{singlepath}$)\Comment{Run algorithm \ref{alg:SAW_sampler} for $t_\text{singlepath}$ steps to sample from the conditional \eqref{eq:conditional_gibbs}}
\For{$k=1,\dots,\text{length($\pi^\nu$)}-1$}
\State $e\gets [\pi^\nu[k],\pi^\nu[k+1]]$
\State Gflow.weight($e$)$\gets$ Gflow.weight($e$)+1\Comment{adding back path $\nu$ to the graph}
\State Gcost.weight($e$)$\gets\phi(\text{Gflow.weight($e$)}+1)-\phi(\text{Gflow.weight($e$)})$
\EndFor
\EndFor
\State newH$\gets 0$ \Comment{measuring energy after all paths are updated}
\For{$e \in G.$edges}
\State newH$\gets$newH+Gcost.weight($e$)
\EndFor
\If{H == newH \textbf{and} $\beta==\infty$}
\State flagconv$\gets$True \Comment{If $\beta=\infty$ and  $H(\pmb I)$ does not change once all the paths are updated, then the algorithm has converged}
\EndIf
\State H $\gets$ newH
\EndWhile
\end{algorithmic}
\end{algorithm}
The function 'weightMatrix(Gcost)' returns an $N\times N$ (with $N$ being the number of vertices) matrix $W$ where $W_{ij}=$Gcost.weight($[i,j]$) if edge $i\to j$ exists in $G$ and $W_{ij}=\infty$ otherwise. 
Notice that the pseudocode of this algorithm is almost identical to that of the greedy algorithm \ref{alg:greedy_algorithm}. The only difference is in lines 28, 29. Erasing line 28 and replacing 29 with '$\pi^\nu\gets$  shortestPath(Gcost, $s^\nu,t^\nu$)' (line 28 in \ref{alg:greedy_algorithm}) one obtains again the greedy algorithm. This relation sheds some light onto the relation between the two algorithms. While the greedy algorithm always picks the shortest path (weighted with $\Delta H$), the simulated annealing is noisy and samples a path from $\propto e^{-\beta\Delta H}$.
\section{RITAP details}
 The RITAP algorithm \ref{alg:relaxed_ITAP_solver} is composed of two parts:
 \begin{itemize}
     \item An algorithm named $\verb|solveTAP|(G,D)$ that outputs a solution to TAP in the form of a collection $\PM=\{\PM_{xy}\}$ with $\PM_{xy}$ defined in \eqref{eq:collection_TAP}.
     \item A projection algorithm that takes as input a TAP solution and outputs an ITAP solution, with the same $\PM^\text{ITAP}$.
 \end{itemize}
 The next two sections focus respectively on the implementation of $\verb|solveTAP|(G,D)$, and on the integer projection step.
 \subsection{Frank-Wolfe as TAP solver}
 \label{sec:FW4TAP}
 To solve TAP we use Frank-Wolfe (FW) algorithm. FW is an iterative algorithm that produces a sequence of edge flows $\{\pmb I(t)\}$ that approaches the optimal edge flows $\pmb I^\star$.
 Each step of FW consists of two sub steps:
 \begin{enumerate}
     \item \textbf{Linearizing $H(\pmb{I}(t))$ near the current point and minimizing the linearized objective}. Denoting with $\Tilde H$ this linear approximation we have $\Tilde H(\pmb I;\pmb I(t))=H(\pmb I(t))+\sum_{e\in E}\phi'(I_e(t))(I_e-I_e(t))$. $\Tilde H(\pmb I;\pmb I(t))$ is minimized with respect to $\pmb I$ under the TAP constraints. Since $\Tilde H$ is linear in $\pmb I$, at its minimum for each pair of nodes $x,y$, all the flow from $x$ to $y$ is routed through a shortest path in the graph whose edges are weighted with $\phi'(I_e(t))$. The new flow $\Tilde{\pmb I}(t)$ obtained from this minimization is called the all-or-nothing flow, since between every pair of nodes all the flow is routed through a single path. 
     \item \textbf{Step towards the all-or-nothing solution}:
     We move from $\pmb I(t)$ towards the all-or-nothing solution with a step size $\eta(t)$.
      $\pmb I(t+1)=(1-\eta(t))\pmb I(t)+\eta(t)\pmb{\Tilde{\pmb I}}(t)$.
      The step size is taken to be decreasing in $t$ with an appropriate schedule that guarantees convergence.
 \end{enumerate}

When $H(\pmb I)$ is convex, FW gives, at each step, a lower bound $H_\text{LB}(t)$ to the optimal energy. This allows to assess convergence of the algorithm by measuring $H(\pmb{I}(t))-H_\text{LB}(t)$. 

The following is a one line proof of the fact that $H_\text{LB}(t)=\min_{\Tilde{\pmb I}} \Tilde{H}(\pmb I,\pmb I(t))$ \footnote{where the minimization over $\Tilde {\pmb I}$ is carried out under the TAP constraints} is a lower bound to $H(\pmb I^\star)$ 
\begin{equation}
    H(\pmb I^\star)\geq H(\pmb I(t))+\sum_e \phi'(I_e)(I_e^\star-I_e(t))\geq \min_{ \Tilde{\pmb I}} H(\pmb I(t))+\sum_e \phi'(I_e)(\Tilde I_e-I_e(t))=\min_{\Tilde{\pmb I}} \Tilde{H}(\pmb I,\pmb I(t)),
\end{equation}
where the minimization over $\Tilde {\pmb I}$ is carried out under the TAP constraints, and the first inequality holds by convexity. Notice that $\Tilde{H}(\pmb I,\pmb I(t))$ is minimized at every time step, therefore the bound is readily available. 

In our experiments we run FW with the following specifications:
\begin{itemize}
    \item The edge flows at initialization $\pmb I(t=0)$ are obtained by computing the shortest paths (in the topological distance on $G$) between every pair of nodes $x,y$ with positive demand. Denoting with $\pi_{xy}$ the a topological shortest path between $x$ and $y$, we have $ I(t=0)_e=\sum_{x,y\in (V\times V)} D_{xy} \mathbb I[e\in\pi_{xy}]$.
    \item Supposing we run FW for $t_\text{max}$ steps, our step size schedule is $ \eta(t)=t^{-\alpha (t)}$ with 
    \begin{equation}
    \label{eq:step_size_schedule}
        \alpha(t)=\begin{cases}
    \frac{1}{2}(1+\frac{t}{t_\text{switch}}) & \text{if  $t\leq t_\text{switch}$} \\
    1.35 &\text{if  $t>t_\text{switch}$},
  \end{cases}
    \end{equation}
    with $t_\text{switch}=0.9\, t_\text{max}$. We picked this schedule since we observed it works well in practice. The idea behind it is to start with a step size $\eta(t)=1/\sqrt{t}$ and gradually move towards a rate $\eta(t)=1/t$. The last 10\% of iterations we decrease to $t^{-1.35}$, this has the effect of obtaining a better lower bound to the TAP optimal energy. Without worrying about the lower bound the schedule $\alpha(t)=\frac{1}{2}(1+\frac{t}{t_\text{max}})$ is found to work very well and leads to $H(\pmb I(t))/H(\pmb I^\star)-1\approx 1/t^3$.
    \item The classic FW algorithm outputs only the optimal edge flows, and does not return the optimal paths. We modify FW to keep track, throughout the iterations, of the paths on which flow is allocated. The paths together with the corresponding amount of flow routed over  them are returned.
\end{itemize}
Notice that the FW algorithm does not depend explicitly on $M$, the number of OD pairs. $M$ is instead reflected in the magnitude of the entries of $D$. Thanks to this independence FW is particularly efficient for very large $M$.
\subsection{ITAP projection}
Once we have a TAP solution we must project it onto a set of path flows satisfying the integrality constraint. RITAP as presented in \ref{alg:relaxed_ITAP_solver} implements one such procedure outputting the paths $\{\pi^\mu\}_{\mu\in[M]}$. Crucially this procedure complexity is linear in $M$. Below in \ref{alg:relaxed_ITAP_solver_indep_M} we provide a different algorithm that achieves the same result, but whose dependence on the number of paths goes like $\min(M,N^2)$.

\begin{algorithm}[H]
\caption{Relaxed ITAP solver (RITAP)} 
\label{alg:relaxed_ITAP_solver_indep_M}
\begin{algorithmic}[1]
\State\textbf{Input:} $G$ graph over which paths are sampled, $D$ demand matrix. 
\State\textbf{Output:} $\PM^\text{ITAP}$ collection of sets in the form \ref{eq:collection_TAP} containing the paths flows satisfying ITAP constraints.
\State $\PM \gets \verb|solveTAP|(G,D)$. \Comment{Solving TAP}
\State $\PM^\text{ITAP}\gets \text{copy}(\PM)$
\State $\PM^\text{tmp}\gets \text{copy}(\PM)$

\For{$x\in G.$nodes}
\For{$y\in G.$nodes}
\If{$D_{xy}>0$}
\State allocatedFlow$\gets 0$
\For{$a\in [1,2,\dots, R_{xy}]$}
\State $\PM^\text{ITAP}_{xy}.h[a]\gets \lfloor\PM^\text{ITAP}_{xy}.h[a]\rfloor$\Comment{Keeping only the integer part of the flow on each path}
\State allocatedFlow$\gets$ allocatedFlow $+\PM^\text{ITAP}_{xy}.h[a]$
\State $\PM^\text{tmp}_{xy}.h[a]\gets \PM^\text{tmp}_{xy}.h[a]$ mod 1 \Comment{Keeping only the fractional part}
\EndFor
\State sortedFracFlows, permIdx $\gets$ sortTogether($\PM^\text{tmp}_{xy}.h, [1,2,\dots, R_{xy}]$) \Comment{Sort by fractional part}
\For{$k\in[0,1,\dots, D_{xy}-\text{allocatedFlow}]$} \Comment{Allocating the remaining flow in decreasing order of fractional flow}
\State $\PM^\text{ITAP}_{xy}.h[\text{permIdx[$k$]}]\gets \PM^\text{ITAP}_{xy}.h[\text{permIdx[$k$]}]+1$
\EndFor 
\EndIf
\EndFor
\EndFor
\State \textbf{return } $\PM^\text{ITAP}$
\end{algorithmic}
\end{algorithm}
the function sortTogether takes two lists as arguments, it sorts the first list in descending order, and applies the same permutation to the second list. For example SortTogether($[3,4,1,5,2], [a,b,c,d,e]$)=$[5,4,3,2,1],[d,b,a,e,c]$.
Basically this algorithm builds $\PM^\text{ITAP}$ by first keeping the integer part of the TAP path flows and then adding the remaining integer units of flow, on the paths with the highest fractional part.
 
\label{app:ITAP_relaxed}

\section{Additional numerical experiments in the regime $M=\Theta(N^2)$}
\label{app:num_exp_M=N^2}
In this section we present some complementary results to those shown in section \ref{sec:integ_relax_solutions} and \ref{sec:real_data_exp}. We consider the setting of section \ref{sec:integ_relax_solutions}, consisting of undirected RRG with $d=3$, OD pairs sampled uniformly at random and $\phi(x)=x^2$. We run experiments each time changing one of these factors and leaving the others unaltered. Below is a list of the modified settings.
\begin{itemize}
    \item Considering RRG with higher degrees, instead of $d=3$
    \item Considering real world traffic networks instead of RRGs
    \item Considering a realistic Bureau of Public Roads (BPR) nonlinearity instead of $\phi(x)=x^\gamma$
\end{itemize}
\subsection{Varying the degree}
We start by presenting data about RRGs with degree larger than $3$ and fixed size $N=128$. Increasing the degree alters some qualitative characteristics previously observed for $d=3$ in figure \ref{fig:energy_conv_tap_itap}. However crucially the asymptotic behavior is robust to changes in $d$. Panel (a) and (b) show that the gap between the TAP optimal energy and the one outputted by ITAP goes to zero as $\eta$ increases. In panel (c) we observe that the onset of the power law behavior for $1-F_\text{int}$ is displaced when changing $d$, this is potentially because the bound \eqref{eq:bound_Fint_degen} becomes non-vacuous (and hence the $1/\eta$ is expected to arise) when $\eta\geq (K-1)$. Since $K-1$ varies from 0.1 for $d=3$ to 10 for $d=48$, the knee in panel (c) moves accordingly.
Ultimately for large $\eta$ all curves display the  predicted $1/\eta$ rate. Similarly, for all degrees, the average degeneracy $K$ (plotted in (d)) plateaus when $\eta$ increases.
The two-scale (i.e. $\rho,\eta$) behavior also persists when changing $d$. For $\rho=\Theta(1)$ the curves collapse when plotted against $\rho$ and for $\eta\geq 1$ (approximately) one enters in the asymptotic regime. In summary, all of the claims we made about the model still hold when looking as experiments with higher $d$.
\begin{figure}
    \centering
    \includegraphics{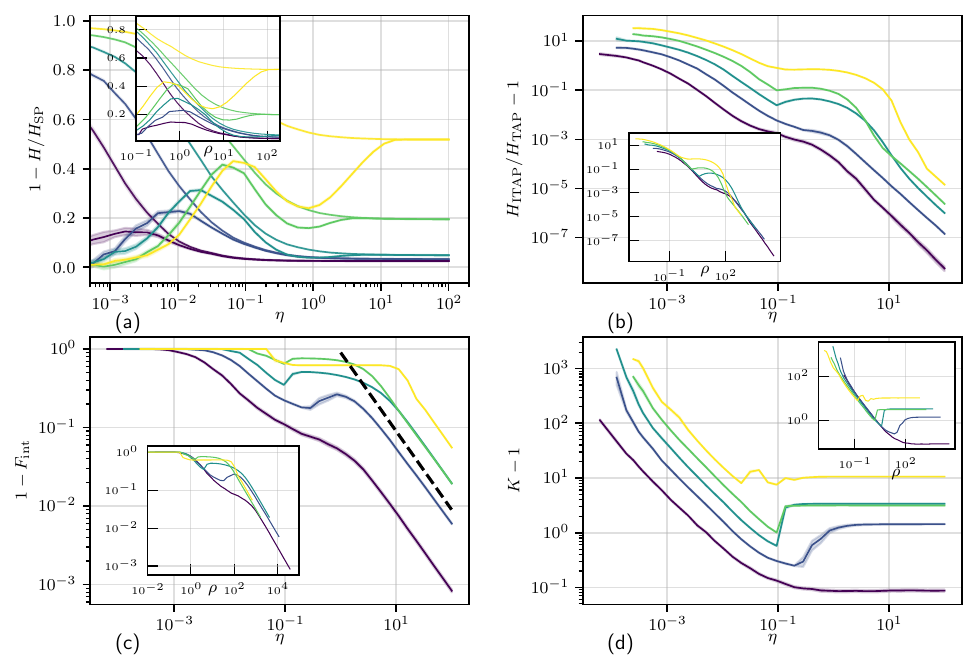}
    \caption{Experiments for $N=128,\,\gamma=2$, on RRG with varying degree. From darker to lighter the degree is $2,6,12,24,48$. Results are averaged over 20 independent realizations of the RRG and of the OD pairs. Shaded areas represent a $\pm1$ standard deviation interval around the mean. Insets show the same quantity plotted in the main panel but as a function of $\rho$. (a) The upper and lower curves correspond respectively to TAP optimal solution and the RITAP solution to ITAP. (b) Ratio of the energies reached respectively by RITAP and TAP. (c) Fraction of integer paths as a function of $\eta$. The dashed line has a slope of $1/\eta$.(d) Average degeneracy as a function of $\eta$}
    \label{fig:TAP_ITAP_conv_N128_varyd}
\end{figure}
\subsection{Varying the nonlinearity}
In traffic the assignment literature the time a user takes to traverse an edge with total traffic $I$ is usually modeled using the Bureau of Public Roads (BPR) functions. 
\begin{equation}
\label{eq:BPR}
    f(I)=t_0\left(1+\alpha\left(\frac{I}{c}\right)^\beta\right),
\end{equation}
where $t_0$ is the crossing time in absence of traffic, $c$ is the capacity of the edge; when the traffic exceeds the capacity the cost quickly rises. In our experiments we take $t_0=1, \alpha=3.8\times 10^{-2}, \beta=4, c=1$.

Since we are interested in minimizing the average cost we set $\phi(x)=x f(x)$\footnote{recall that $\phi$ represents the total cost of one edge. Therefore it is obtained by multiplying the individual cost $f$ times the flow on the edge.}. Results of thee experiments are presented in figure \ref{fig:exp_BPR}. Each line is the result of averaging over 20 independent realizations of the OD pairs and the RRG. We notice an interesting difference with the $\phi(x)=x^\gamma$ case: in panel (a), for small $\eta$ (small $M$), the TAP and RITAP curves are close. This is in contrast with the $\phi(x)=x^\gamma$  case, where the gap between TAP and RITAP was maximum for small $M$. The reason behind this difference lies in the fact that the $xf(x)$ has a linear behavior near zero, therefore when the flow per edge is small, the flow is routed on shortest paths in the topological distance, instead of being split across many paths. In other words, having a nonlinearity that has a nonzero derivative near $I_e=0$ results in TAP support paths that are shortest paths in the topological distance when $M$ is small. Thanks to this fact the degeneracy at small $M$ is quite small compared to the case $\gamma=2$ (see figure \ref{fig:energy_conv_tap_itap} (d)). Apart from this difference in the small $M$ regime, the asymptotic behavior is unchanged. As in the case $\gamma=2$, when $\eta\gg1$, $H_\text{ITAP}/H_\text{TAP}\to 1$, $1-F_\text{int}\to 0$ as $1/\eta$, $K$ approaches a constant value. We conclude that the asymptotic phenomenology is robust to changes in the nonlinearity. One caveat to this statement consists in the choice of the capacity parameter. Our experiments suggest that if $c$ is taken to be large (e.g. a few hundreds), then this introduces an additional scale into the problem. The new scale alters the two regimes ($\rho,\eta$) phenomenology we so far encountered. 
Finally we remark that also in this setting, $\rho$ and $\eta$ represent the relevant parameter to describe the system respectively in the regimes where the flow per edge is $\Theta(1)$ and where $M=\Theta(N^2)$. This is illustrated by the fact that curves collapse when plotted as a function or $\rho$, when $\rho=\Theta(1)$, and as a function of $\eta$ when $\eta\geq1$.

\begin{figure}
    \centering
\includegraphics{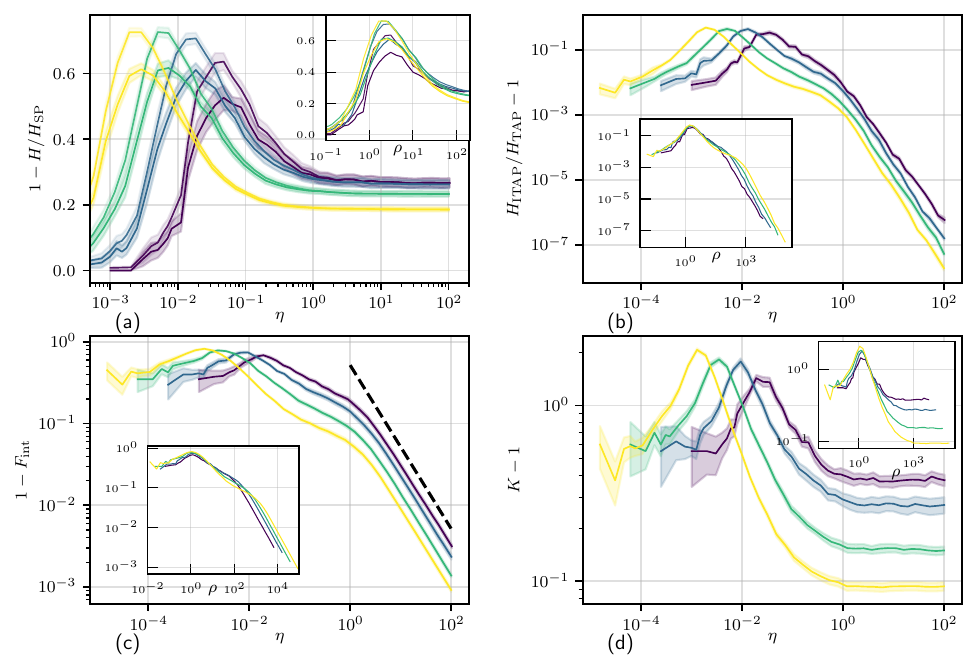}
    \caption{Experiments on RRG with $d=3$, OD pairs picked uniformly at random, and nonlinearity of the form \eqref{eq:BPR} with $\alpha=0.038,\, \beta=4,\, t_0=1,\, c =1$. (a) Relative energy gain with respect to the shortest path routing. The upper curves  correspond to TAP optimality, the lower ones to the RITAP solution. (b) Ratio between the energy reached by RITAP and the TAP optimal energy. (c) Fraction of integer paths as a function of $\eta$. The dashed line has a slope of $1/\eta$.(d) average degeneracy of support paths as a function of $\eta$}
    \label{fig:exp_BPR}
\end{figure}

\section{Numerical experiments details}
\label{app:num_exp_details} 
In this appendix we describe the details of the numerical experiments.
\subsection{$\rho$ approximates the average flow per edge}
We start by presenting data that confirm that $\rho=2M\log N/(N d \log d)$ is a good approximation to the flow per edge in the random model ($G\sim$ RRG, OD pairs sampled uniformly at random). In principle since every algorithm outputs a different set of paths, we would need to prove that this approximation holds for every algorithm considered. In practice however, there is little difference between the algorithms, therefore we only show data for RITAP. Given the RITAP edge flows $\pmb I_\text{RITAP}$ we define the average flow per edge to be $\langle I\rangle_\text{RITAP}\coloneqq \frac{1}{|E|}\sum_{e\in E} I^\star_e$.
Figures \ref{fig:rho_approx_avg_flow_per_edge} and \ref{fig:rho_approx_avg_flow_per_edge_gamma05} illustrate the agreement between the two quantities by plotting $\langle I\rangle_\text{RITAP}/\rho$ as a function of $\rho$ respectively for $\gamma=2$ and $\gamma=1/2$. The left panels consider the case of fixed $d=3$ and changing $N$, while in the right panels we fix $N=128$ and change $d$. The deviation of $\langle I\rangle_\text{RITAP}/\rho$ from one is always smaller than 30\%, for $\gamma=2$ and smaller than $3$ for $\gamma=1/2$. While a factor three in the discrepancy between may seem a lot, we remark that $\rho$ has a quite simple analytical expression that captures the dependence of $\langle I\rangle_\text{RITAP}$ on $M,N,d$ for a wide range of these variables.
\begin{figure}
    \centering
\includegraphics{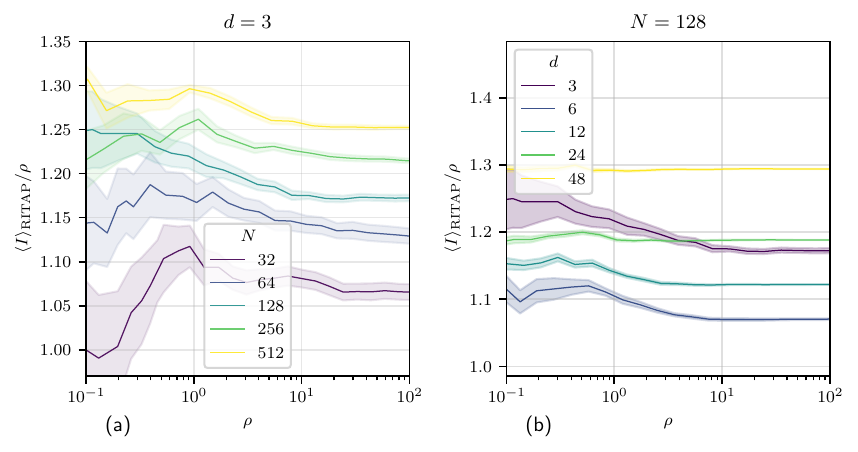}
    \caption{ $\langle I\rangle_\text{RITAP}/\rho$ as a function of $\rho$ for $\gamma=2$. Experiments are done on RRGs with OD pairs picked uniformly at random, . Results are averages over 10 instances for $N=512$ and over 20 instances for other $N$s. Both when varying the degree and the the size, $\rho$ is within 30\% of the average flow per edge.
    (a) degree fixed to 3 and varying $N$. (b) $N$ fixed to $128$ and varying degree. }
    \label{fig:rho_approx_avg_flow_per_edge}
\end{figure}
\begin{figure}
    \centering
\includegraphics{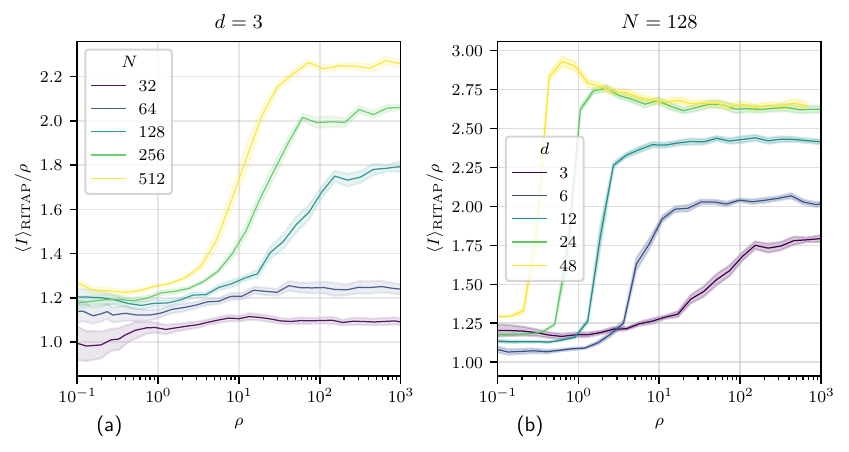}
    \caption{ $\langle I\rangle_\text{RITAP}/\rho$ as a function of $\rho$ for $\gamma=1/2$. Experiments are done on RRGs with OD pairs picked uniformly at random. Results are averages over 10 instances for $N=512$ and over 20 instances for other $N$s. Both when varying the degree and the the size, $\rho$ is within three times the average flow per edge.
    (a) degree fixed to 3 and varying $N$. (b) $N$ fixed to $128$ and varying degree. }
    \label{fig:rho_approx_avg_flow_per_edge_gamma05}
\end{figure}

\subsection{Experiments in section \ref{sec:integ_relax_solutions} with $M=\Theta (N^2)$}
The results shown in section \ref{sec:integ_relax_solutions} are obtained by averaging over 10 independent instances of graphs and demand matrices, for $N=512$ and over 20 independent instances for other sizes. A common concern, when running the Frank-Wolfe algorithm, is the proximity to the global optimum. Denote with $H_\text{LB}$ the lower bound to $H(\pmb I^\star)$ obtained for $t=t_\text{max}$, and with $H_\text{FW}$ the final energy reached by the FW algorithm. The quantity $1-H_\mathrm{LB}/H_\mathrm{FW}$, is an upper bound to $1-H(\pmb I^\star)/H_\mathrm{FW}$ which in turn represents the relative energy gap between the FW solution and the global optimum. Figure \ref{fig:rel:opt_gap_FW} plots this quantity as a function of $\eta$. Notice that in our experiments we used a number of steps dependent of $M, N,d$, therefore these curves do not to show how $M,N,d$ influence convergence. Instead they indicate that for basically all experiments shown in the main text $(d=3)$ the relative error in the optimal energy is less than $10^{-7}$.
\begin{figure}
    \centering
    \includegraphics{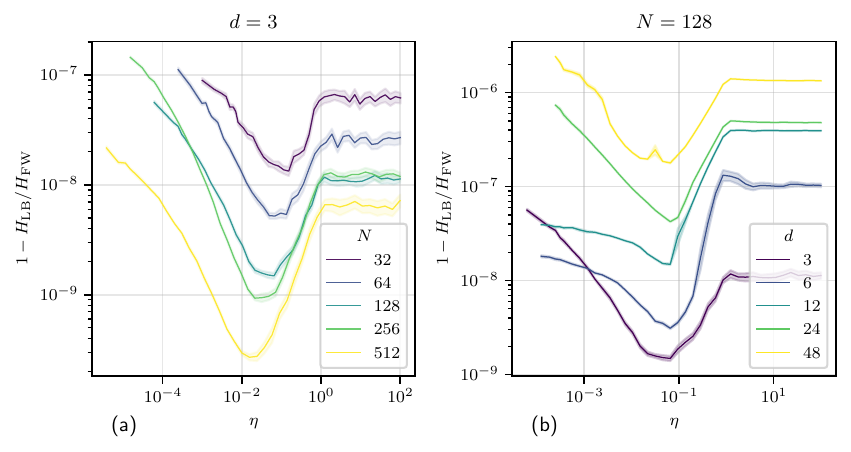}
    \caption{Upper bound to the relative error in the optimal TAP energy. The error originates from the imperfect convergence of the Frank-Wolfe algorithm to the TAP global minimum. All curves are averages over 20 independent realizations of the RRG and the demand matrix (where OD pairs are picked uniformly at random). 
    (a) Fixing $d=3$ and varying the size of the graph. (b) Fixing $N=128$ and changing the degree.}
    \label{fig:rel:opt_gap_FW}
\end{figure}
\subsection{Computation of $F_\text{int}$ and $K$}
To compute the fraction of integer paths in a TAP solution, we need to determine wen the flow on a path in $\PM$ (see \ref{eq:collection_TAP} for notation) has integer flow. To do so we check if the flow is within an integer number within a tolerance of $10^{-4}$. We remark that the fraction of integer paths is quite insensitive to this threshold as long as it is kept between $10^{-6}$ and $10^{-3}$.

To compute $K$ we compute the set of support paths using Wardrop's second principle \ref{prop:wardrop_2_principle}. To compute $\mathcal{S}_{xy}$ we use the following procedure. We set a relative tolerance threshold $\epsilon$ and consider the graph with the edges weighted with $\phi'( I^\text{FW}_e)$, where $\pmb I^\text{FW}$ is the flow on edge $e$ after running Frank-Wolfe. Let $c$ be the cost of the shortest path between $x$ and $y$ in such graph. We use Yen's algorithm \cite{yen1971finding} to compute all the deviations from the shortest path in order of increasing cost, stopping when a path of cost $(1+\epsilon)c$ is generated. All paths generated are considered support paths. In practice we use $\epsilon=10^{-5}$. We find that our estimate of $K$ is quite insensitive to $\epsilon$ in the range $[10^{-6}, 10^{-3}]$.

\subsection{Experiments in section \ref{section_results_comparison_algos_small_M}}
In this appendix we specify the hyperparameters used in the experiments with $\rho=O(1)$.
\begin{enumerate}
    \item \textbf{Greedy algorithm:} looking at the pseudocode of the algorithm \ref{alg:greedy_algorithm}, we see that we must only specify the paths at initialization. We choose these to be the shortest paths in the topological distance on the graph. 
    \item \textbf{Simulated annealing} requires the choice of a value of the starting inverse temperature $\beta_{min}$, to be plugged in the annealing schedule \eqref{eq:annealing_schedule}. For a very high value of this parameter simulated annealing would give the same results as the greedy algorithm. Conversely, a very low value (very high temperature) for $\beta_{min}$ would be responsible for the algorithm is very slow as it spends a lot of time at high temperature. We have observed numerically that a value around $\beta_{min}=20$ is suitable for the sizes and degrees of graphs considered.  
    The number of annealing steps $T_\text{ann}$ (see  \eqref{eq:annealing_schedule}) has been set to 30. After the annealing schedule has terminated, the value of $\beta$ is set to $\infty$ and 20 more steps (in each one updating all the paths) are carried out.
    When running the simulated annealing one must sample SAWs with fixed endpoints. To do so, algorithm \ref{alg:SAW_sampler} is used. This algorithm is itself an MCMC and the number of steps to take every time we sample a SAW is a hyperparameter which is set to 2.
    \item \textbf{RITAP:} The RITAP algorithm has several hyperparameters. The maximum number of iterations is set to $t_{max}=10^4$.
    
    The simulation is stopped before if a satisfactory convergence is reached before $t_\text{max}$. We set the relative tolerance to $\text{rtol}=10^{-9}$. The algorithm halts if $H(t)-H(t+1)/H(t)<\text{rtol}\, \eta(t)$, where $\eta(t)$ is the step size at time $t$.
    In the case $\gamma=2$ the step size schedule introduced in \eqref{eq:step_size_schedule}.
    In the case $\gamma=1/2$ instead we used a different step size schedule where all the flow is shifted at every step onto the all-or-nothing solution (see \ref{sec:FW4TAP}). This corresponds to setting $\eta(t)=1$ for all $t$. This second schedule is guaranteed to converge to a local minimum of TAP in the concave case (see \ref{sec:TAP=ITAP_attractive}). Since the minimum is not unique, in principle in the concave case the step size schedule can affect the depth of the minimum reached by RITAP. To test this possibility we compare the energy gains of the two schedules (\eqref{eq:step_size_schedule} and $\eta=1$) in figure \ref{fig:compare_schedules_FW}. The first schedule corresponds to  \eqref{eq:step_size_schedule}, while the second schedule to $\eta=1$. The fact that the two reach very similar energies indicates that the schedule does not influence much the final energy. Our choice to use $\eta=1$ in the concave case is then due to its simplicity.
    \item \textbf{CBP:} The CBP algorithm takes as a hyperparameter the maximum number of iterations $MaxIter$ that has been fixed to $MaxIter=20000 N$ where $N$ is the number of nodes in the graph. For instance, for an $N=200$ graph and $\gamma=2$, the number of iterations grows linearly and only reaches a value of $2.5\cdot10^6$ for the highest value of $\rho$ considered in our plots (while $MaxIter=4\cdot10^6$). Another hyperparameter is the number of iterations in which the message don't change that we wait before declaring convergence. It corresponds to the quantity $EqTime$ in \ref{alg:CBP_pseudocode} and it is set to $EqTime=100 N$. Finally, the convergence threshold $\epsilon$ (\ref{alg:CBP_pseudocode}) is fixed to $10^{-10}$.
\end{enumerate}

\begin{figure}[!ht]
    \centering
    \includegraphics[scale=1]{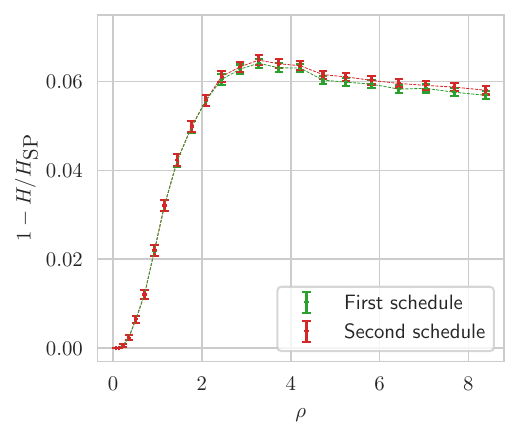}
    \caption{Relative energy difference between the paths using each of the schedules described for RITAP. The higher the better. Results are averages over 200 instances of ITAP, with $d=3$, $N=200$, $\gamma=1/2$. First schedule corresponds to \eqref{eq:step_size_schedule}, second schedule to $\eta(t)=1$, for all $t$.}
    \label{fig:compare_schedules_FW}
\end{figure}

\section{Proofs}
\label{app:proofs}
\subsection{Wardrop's second principle}
The following is a proof of Wardrop's second principle \ref{prop:wardrop_2_principle}.

\begin{proposition}[Wardrop's second principle]
\label{prop:wardrop_2_principle_app}
Let $G,D$ be a TAP instance and $\{h_{xy}^{\star a}\}_{(x,y)\in V\times V, a\in [|\Pi_{xy}|]}$ be any path flows at optimality. Also let $\pmb{I^\star}$ be the edge flows at optimality. \textbf{Then} for every $(x,y)\in V\times V$, for every $a\in[|\Pi_{xy}|]$, a necessary condition for $h_{xy}^{\star a}$ to be strictly positive is that
\begin{equation}
    \sum_{e\in \pi_{xy}^a} \phi'(I_e^\star)=\min_{\pi\in\Pi_{xy}}  \sum_{e\in \pi} \phi'(I_e^\star)
\end{equation}
\end{proposition}
\begin{proof}

We proceed by writing down the KKT conditions for TAP \ref{eq:TAP_ITAP_formal}.  
The TAP Lagrangian written in terms of path flows reads
\begin{align}
\label{eq:lagrangian_TAP}
    \mathcal L\left(\pmb h,\pmb \lambda,\pmb \mu\right)=\sum_{e\in E}\phi\left( I_e\right)+\sum_{x,y}\lambda_{xy}\left(\sum_{\pi\in\Pi_{xy}}h_{xy}^\pi-D_{xy}\right)+\sum_{x,y}\sum_{\pi\in\Pi_{xy}}\mu_{xy}^\pi h_{xy}^\pi,
\end{align}
where in the first term we used the identity $I_e=\sum_{\pi\in\Pi_{xy}} h_{xy}^\pi\mathbb I[e\in \pi]$.
The second term corresponds to the constraint that the path flows satisfy the demand for every pair of nodes, and the third term enforces the positivity of all flows.

The stationarity with respect to $\pmb h$ imposes
\begin{equation}
    \sum_{e\in \pi}\phi'(I_e)+\lambda_{xy}+\mu_{xy}^\pi=0 \quad \forall (x,y),\,\forall \pi\in \Pi_{xy}.
\end{equation}
The primal feasibility imposes 
\begin{align}
    &\sum_{\pi\in\Pi_{xy}} h_{xy}^\pi=D_{xy},\quad \forall (x,y)\\
    & h_{xy}^\pi\geq0, \quad \forall (x,y), \,\forall \pi\in\Pi_{xy}.
\end{align}
The dual feasibility imposes $\mu_{xy}^\pi\geq 0, $ $\forall (x,y),\,\forall \pi\in\Pi_{xy}$.
And finally the complementary slackness imposes $\mu_{xy}^\pi h_{xy}^\pi=0,$ $\forall (x,y),\,\forall \pi\in \Pi_{xy}$.

Fix an arbitrary pair of nodes $(x,y)$, and consider the paths that carry positive flow $(h_{xy}^\pi>0)$ at optimality: these are the support paths $\mathcal S_{xy}$. The complementary slackness imposes that $\mu_{xy}^\pi=0$ for all $\pi\in\mathcal S_{xy}$.
Plugging this back into the stationarity condition we get that all paths carrying a positive flow satisfy 
\begin{equation}
\label{eq:same_traveltime_kkt_wardrop2}
    \sum_{e\in \pi}\phi'(I_e)=-\lambda_{xy} \quad \forall \pi\in \mathcal S_{xy}.
\end{equation}
Equation \ref{eq:same_traveltime_kkt_wardrop2} basically tells us that all support paths in $\mathcal S_{xy}$ must indeed have the same $\sum_{e\in \pi}\phi'(I_e)$. Using second order optimality conditions it can also be proved that $\sum_{e\in \pi}\phi'(I_e)$ is minimal on support paths. 
\end{proof}
\subsection{Using $K$ to bound $F_\text{int}$}
In section \ref{eq:bound_Fint_degen} we stated that $  F_\text{int}\geq1-\frac{N(N-1)}{M}\left(K-1\right)$.
To prove it, we shall now prove the following result, which is slightly stronger:
\begin{proposition}
\label{prop:bound_Fint_degen}
    Let $F_\text{int}\coloneqq\frac{1}{M}\sum_{x,y \in (V\times V)} \sum_{a\in \mathcal [|S_{xy}|]} \lfloor h_{xy}^a\rfloor$. And suppose $\sum_{a\in \mathcal [|S_{xy}|]}  h_{xy}^a=D_{xy}$ with $D_{xy}\in\N$. \textbf{Then}
    \begin{equation}
\label{eq:bound_Fint_degen_app}
     F_\text{int}\geq \frac{1}{M}\sum_{x,y\in V\times V}\max(0,D_{xy}-(|\mathcal S_{xy}|-1)),
\end{equation}
\end{proposition}
\begin{proof}
where $\mathcal S_{xy}$ is the set of support paths from $x$ to $y$. Let us first recall the definition of $F_\text{int}=\frac{1}{M}\sum_{x,y \in (V\times V)} \sum_{a\in [\left|\mathcal S_{xy}\right|]} \lfloor h_{xy}^a\rfloor$. \footnote{notice we exchanged $R_{xy}$ with $|\mathcal S_{xy}|$. This makes no difference.}
Fix a pair of nodes $x,y$ and suppose we must route $D_{xy}$ units of flow from $x$ to $y$. We call $(F_\text{int})_{xy}$ the fraction of integer flow going from $x$ to $y$, i.e., $(F_\text{int})_{xy}\coloneqq \frac{1}{D_{xy}}\sum_{a\in [\left|\mathcal S_{xy}\right|]} \lfloor h_{xy}^a\rfloor$. Using lemma \ref{lemma:max_noninteger_flow} we have 
\begin{align}
    (F_\text{int})_{xy}&\geq 1-\min(D_{xy},\mathcal{S}_{xy}-1)/D_{xy}=1-(D_{xy}+\min(0,\mathcal{S}_{xy}-1-D_{xy}))/D_{xy}=\\&=-\min(0,\mathcal{S}_{xy}-1-D_{xy})/D_{xy}=\max(0,D_{xy}-(\mathcal{S}_{xy}-1))/D_{xy}
\end{align}
Using the expression $F_\text{int}=\frac{1}{M}\sum_{x,y} D_{xy} (F_\text{int})_{xy}$, we complete the proof.
\end{proof}
From Proposition \ref{prop:bound_Fint_degen} we obtain the original bound by relaxing the maximum $\max(0,D_{xy}-(|\mathcal S_{xy}|-1))\geq D_{xy}-(|\mathcal S_{xy}|-1)$. By then taking the sum over $x,y; \, x\neq y$ and using the definition of $K$,  we obtain \eqref{eq:bound_Fint_degen}.
We now prove the auxiliary lemma we used in the proof.

\begin{lemma}
\label{lemma:max_noninteger_flow}
Let $S,D$ be two positive integers, and consider the following optimization problem
\begin{equation}
\text{\textbf{max fractional flow}}=\max_{\pmb h: \sum_{a=1}^S h_a=D}  \sum_{a=1}^S\left(h_a-\lfloor h_a\rfloor\right).
\end{equation}
Then $\text{\textbf{max fractional flow}}=\min (D,S-1)$.
\end{lemma}
\begin{proof}
    Since  $\sum_{a=1}^S h_a=D$ it is also true that $\sum_{a=1}^S\left(h_a-\lfloor h_a\rfloor\right)\leq D$.
    Now we shall prove that $\sum_{a=1}^S\left(h_a-\lfloor h_a\rfloor\right)\leq S-1$.
    To do so, notice that $\sum_{a=1}^S\left(h_a-\lfloor h_a\rfloor\right)<S$, since every term in the sum is strictly smaller than one. Moreover $\sum_{a=1}^S\left(h_a-\lfloor h_a\rfloor\right)$ must be integer. This is because both $\sum_{a=1}^S h_a=D$ and $\sum_{a=1}^S\lfloor h_a\rfloor$ are integer. This implies that $\sum_{a=1}^S\left(h_a-\lfloor h_a\rfloor\right)\leq S-1$. To show that ti can be equal to $S-1$, consider Let $h_a=(S-1)/S$ for all $a$. For this $\pmb h$ the sum of fractional parts is indeed $S-1$.
\end{proof}
\subsection{Proof that $F_\text{int}\to 1$ implies $H_\text{ITAP}/H_\text{TAP}\to 1$}

In this section we establish that if the traction of integer paths approaches one, then the ratio $H_\text{ITAP}/H_\text{TAP}$ also does.

\begin{proposition}
    Let $G$ be an undirected graph with $N$ nodes and $|E|$ edges, and let $D$ a demand matrix. Let $\phi(x)=x^p$, with $p\geq 1$ be the nonlinearity in TAP. 
    Consider $H_\text{TAP},H_\text{ITAP}$ the energies respectively reached by TAP and RITAP when run on $G,D$. Furthermore let $F_\text{int}$ be the fraction of integer paths in TAP as defined in \ref{sec:integ_relax_solutions}.\textbf{Then}
    \begin{equation}
    \label{eq:bound_H_ITAP_TAP_Fint}
        \frac{H_\text{ITAP}}{H_\text{TAP}}\leq \left[1+2N|E|^{1-1/p}(1-F_\text{int})\right]^p
    \end{equation}
    
\end{proposition}
\begin{proof}
    Let $\pmb I^\text{TAP},\pmb I^\text{ITAP}$ be respectively the TAP edge flows at optimality and the ITAP flows outputted by RITAP.
   First we introduce the integer TAP flow $\pmb I^\text{INT}$, given by truncating the decimal part of all path flows in TAP. Mathematically
   $I^\text{INT}_e=\sum_{x,y}\sum_{a\in \PM^\text{TAP}_{xy}} \lfloor h_{xy}^a \rfloor \mathbb I[e\in \pi_{xy}^a]$. 
   $\pmb I^\text{INT}$ can be thought of as a middle point between  $\pmb I^\text{TAP}$ and  $\pmb I^\text{ITAP}$, it fact RITAP proceeds by keeping the flow $\pmb I^\text{INT}$ and reassigning the rest (the fractional part).
   
   Our goal is to upper bound $\norm{\pmb I^\text{ITAP}-\pmb I^\text{TAP}}$. Using the triangular inequality we have
   \begin{equation}
   \label{eq:TAP_ITAP_triang_ineq}
       \norm{\pmb I^\text{ITAP}-\pmb I^\text{TAP}}=\norm{\pmb I^\text{ITAP}-\pmb I^\text{INT}+\pmb I^\text{INT}-\pmb I^\text{TAP}}\leq \norm{\pmb I^\text{ITAP}-\pmb I^\text{INT}}+\norm{\pmb I^\text{TAP}-\pmb I^\text{INT}}
   \end{equation}
   We shall now bound each of the two terms in $\ell_1$ norm:
   
\begin{align}
\label{eq:bound_l1_norm_I_diff}
    &\norm{\pmb I^\text{TAP}-\pmb I^\text{INT}}_1=\sum_{e\in E}\sum_{x,y}\sum_{a\in \PM^\text{TAP}_{xy}}\left( h_{xy}^a-\lfloor h_{xy}^a \rfloor\right)\mathbb I[e\in\pi_{xy}^a]=\sum_{x,y}\sum_{a\in \PM^\text{TAP}_{xy}}\left( h_{xy}^a-\lfloor h_{xy}^a \rfloor\right) \text{length}(\pi_{xy}^a)\\&\leq N\sum_{x,y}\sum_{a\in \PM^\text{TAP}_{xy}}\left( h_{xy}^a-\lfloor h_{xy}^a \rfloor\right)=NM(1-F_\text{int})
\end{align}
In the inequality we simply used $\text{length}(\pi)\leq N$ for all paths.
By replacing $\pmb I_\text{TAP}$ with $\pmb I_\text{ITAP}$ the series of inequalities is still valid, so we also have $\norm{\pmb I^\text{ITAP}-\pmb I^\text{INT}}_1\leq NM(1-F_\text{int})$. And finally from \ref{eq:TAP_ITAP_triang_ineq} we obtain $\norm{\pmb I^\text{ITAP}-\pmb I^\text{TAP}}\leq 2 NM(1-F_\text{int})$.

Now write 
$H(\pmb I^\text{ITAP})-H(\pmb I^\text{TAP})=\norm{\pmb I^\text{ITAP}}_p^p-\norm{\pmb I^\text{TAP}}_p^p$. Our goal is to find an upper bound to this difference in terms of $\norm{\pmb I^\text{ITAP}-\pmb I^\text{TAP}}$.

Once again we employ the triangular inequality applied to the triangle with sides $\pmb I^\text{ITAP},\pmb I^\text{TAP},\pmb I^\text{ITAP}-\pmb I^\text{TAP}$.  This gives $\norm{\pmb I^\text{ITAP}}_p\leq \norm{\pmb I^\text{TAP}}_p+\norm{\pmb I^\text{ITAP}-\pmb I^\text{TAP}}_p$. 
Elevating to the power $p$ we get 
\begin{equation}
    \norm{\pmb I^\text{ITAP}}_p^p\leq \left[\norm{\pmb I^\text{TAP}}_p+\norm{\pmb I^\text{ITAP}-\pmb I^\text{TAP}}_p\right]^p
\end{equation}

And therefore the ratio between ITAP and TAP energy can be bounded as
\begin{align}
   \frac{ \norm{\pmb I^\text{ITAP}}_p^p}{\norm{\pmb I^\text{TAP}}_p^p}
   \leq \frac{1}{\norm{\pmb I^\text{TAP}}_p^p}\left[\norm{\pmb I^\text{TAP}}_p+\norm{\pmb I^\text{ITAP}-\pmb I^\text{TAP}}_p\right]^p=\left[1+ \frac{\norm{\pmb I^\text{ITAP}-\pmb I^\text{TAP}}_p}{\norm{\pmb I^\text{TAP}}_p}\right]^p\leq \left[1+ \frac{\norm{\pmb I^\text{ITAP}-\pmb I^\text{TAP}}_1}{\norm{\pmb I^\text{TAP}}_p}\right]^p,
\end{align}
where in the second inequality we used that $\norm{v}_p\leq \norm{v}_1$, if $p\geq1$.
As a last ingredient we use lemma \ref{lemma:lower_bound_H_TAP} to \textit{lower} bound $\norm{I^\text{TAP}}_p^p\geq |E|(M/|E|)^p$.
This way the bound becomes
\begin{equation}
    \frac{ \norm{\pmb I^\text{ITAP}}_p^p}{\norm{\pmb I^\text{TAP}}_p^p}
   \leq \left[1+ \frac{2NM(1-F_\text{int})}{(M/|E|^{1-1/p})}\right]^p= \left[1+ 2N|E|^{1-1/p}(1-F_\text{int})\right]^p
\end{equation}
\end{proof}
Let us comment briefly on the sources of looseness in the bound. First there is the fact that we are upper bounding the length of a path with $N$ in \eqref{eq:bound_l1_norm_I_diff}. In an RRG for example the typical length would be $\log N$. Second comes the use of the $\ell_1$ norm to upper bound the $\ell_p$ norm. Third, in lemma \ref{lemma:lower_bound_H_TAP} we lower bound the length of a path with one. Again, for example in an RRG a typical value would be $\log N$. 

In an RRG, one can for example hope to tighten the bound by taking into account the first and third points above: this would result in eliminating $N$ on the right hand side of \eqref{eq:bound_H_ITAP_TAP_Fint}.
\begin{lemma}
\label{lemma:lower_bound_H_TAP}
    Let $G,D$ be a TAP instance with $\phi$ convex, and indicate with $H_\text{TAP}$ the TAP optimal energy. \textbf{Then}
    \begin{equation}
        H_\text{TAP}\geq |E|\phi\left(\frac{M}{|E|}\right)
    \end{equation}
\end{lemma}
\begin{proof}
Let $\pmb I$ be the TAP flows at optimality, and $M\coloneqq \sum_{x,y} D_{xy}$.
To obtain the bound we start from the observation that $\sum_{e\in E} I_e\geq M$. Consider the constraints $I_e=\sum_{(x,y)\in(V\times V)} \sum_{a\in[|\Pi_{xy}|]} h^a_{xy} \mathbb I[e\in\pi^a_{xy}] \quad \forall e\in E$ in \ref{eq:TAP_ITAP_formal}. Taking the sum over $e\in E$ of these constraint gives the identity
 \begin{equation}
        \sum_{e\in E} I_e=\sum_{e\in E}\sum_{(x,y)\in(V\times V)}\sum_{a\in[|\Pi_{xy}|]} h^a_{xy} \I[e\in\pi^a_{xy}]=\sum_{(x,y)\in(V\times V)}\sum_{a\in[|\Pi_{xy}|]} h^a_{xy} \;\text{length}(\pi_{xy}^a)
    \end{equation}
Since $\text{length}(\pi_{xy}^a)\geq1$ we obtain 
$\sum_{e\in E} I_e\geq \sum_{(x,y)\in(V\times V)}\sum_{a\in[|\Pi_{xy}|]} h^a_{xy}=M$.
Using Jensen's inequality we get
\begin{equation}
    H_\text{TAP}/|E|= \frac{1}{|E|}\sum_e \phi(I_e)\geq \phi\left(\frac{\sum_{e\in |} I_e}{|E|}\right)\geq\phi\left(\frac{M}{|E|}\right),
\end{equation}
    which concludes the proof.
\end{proof}

\subsection{Proof that doubling the demand matrix does not change the support paths}
Take a TAP instance $G,D$, with $\phi(x)=x^\gamma$. Suppose we know a set of optimal paths and path flows for $G,D$. Consider now the TAP instance $G,\alpha D$, where we multiplied the demand matrix by $\alpha\in (0,\infty)$. Then if we keep the paths of the $G,D$ instance and multiply the path flows by $\alpha$ we obtain an optimal solution to $G,\alpha D$.

Formally we prove the following
\begin{lemma}
\label{lemma:double_D}
    Let $(G,D)$ be a TAP instance with $\phi(x)=x^\gamma$  with $\gamma>1$, Suppose $\PM=\{\PM_{xy}\}_{x,y\in V\times V}$, with 
    \begin{equation}
        \PM_{xy}=\left\{(\pi^a_{xy},h^a_{xy}),\;(x,y)\in(V\times V),\, a\in\{1,\dots,R_{xy}\} \right\}
    \end{equation}
    is an optimal set of paths and path flows for $(G,D)$.
    Let $\alpha\in (0,\infty)$ and consider the TAP instance $(G,\alpha D)$.
    \textbf{Then}
    $\PM^\alpha=\{\PM_{xy}^\alpha\}_{x,y\in V\times V}$, with 
    \begin{equation}
        \PM_{xy}^\alpha=\left\{(\pi^a_{xy},\alpha h^a_{xy}),\;(x,y)\in(V\times V),\, a\in\{1,\dots,R_{xy}\} \right\},
    \end{equation}
    is optimal for the instance $(G,\alpha D)$.
\end{lemma}
\begin{proof}
Consider the KKT conditions for the instance $G,D$. We derived these conditions in the proof of \ref{prop:wardrop_2_principle_app}.
From the Lagrangian \eqref{eq:lagrangian_TAP} we have the following

The stationarity condition imposes
\begin{equation}
\label{eq:stationarity_doubleD}
    \sum_{e\in \pi}\phi'(I_e)+\lambda_{xy}+\mu_{xy}^\pi=0 \quad \forall (x,y),\,\forall \pi\in \Pi_{xy}.
\end{equation}
The primal feasibility imposes 
\begin{align}
\label{eq:primal_feas_doubleD}
    &\sum_{\pi\in\Pi_{xy}} h_{xy}^\pi=D_{xy},\quad \forall (x,y)\\
    & h_{xy}^\pi\geq0, \quad \forall (x,y), \,\forall \pi\in\Pi_{xy}.
\end{align}
The dual feasibility imposes $\mu_{xy}^\pi\geq 0, $ $\forall (x,y),\,\forall \pi\in\Pi_{xy}$.
And finally the complementary slackness imposes $\mu_{xy}^\pi h_{xy}^\pi=0,$ $\forall (x,y),\,\forall \pi\in \Pi_{xy}$.
Suppose the path flows $\{h_{xy}\}_{x,y\in(V\times V), \pi\in \Pi}$ together with the multipliers $\{\lambda_{xy}\}_{x,y \in (V \times V)},\{\mu_{xy}^\pi\}_{(x,y)\in (V\times V), \pi\in \Pi}$ satisfy the KKT equations.

Consider the following path flows $\{\alpha h_{xy}\}_{x,y\in(V\times V), \pi\in \Pi}$ together with the multipliers $\{\alpha^{\gamma-1}\lambda_{xy}\}_{x,y \in (V \times V)}$, \\$\{\alpha^{\gamma-1}\mu_{xy}^\pi\}_{(x,y)\in (V\times V), \pi\in \Pi}$. We now prove that these satisfy the KKT equations for the instance $(G,\alpha D)$.
\begin{itemize}
    \item Let's start from the stationarity condition. The fact that all the path flows are rescaled by $\alpha$ implies that the same happens to the edge flows. Therefore $\phi'(I_e)$ in \eqref{eq:stationarity_doubleD} now becomes $\phi'(\alpha I_e)$. Using $\phi'(x)=\gamma x^{\gamma-1}$ we have that $\phi'(\alpha x)=\alpha^{\gamma-1}\phi'(x)$. Plugging in the new values of the multipliers the equation reads
    \begin{equation}
         \sum_{e\in \pi}\alpha^{\gamma-1}\phi'(\alpha I_e)+\alpha^{\gamma-1}\lambda_{xy}+\alpha^{\gamma-1}\mu_{xy}^\pi=0 \quad \forall (x,y),\,\forall \pi\in \Pi_{xy}.
    \end{equation}
    which is just \eqref{eq:stationarity_doubleD} multiplied by $\alpha^{\gamma-1}$ and therefore is satisfied.
    \item The primal feasibility imposes 
\begin{align}
    &\sum_{\pi\in\Pi_{xy}} \alpha h_{xy}^\pi=\alpha D_{xy},\quad \forall (x,y)\\
    & h_{xy}^\pi\geq0, \quad \forall (x,y), \,\forall \pi\in\Pi_{xy}.
\end{align}
The first equation is just \eqref{eq:primal_feas_doubleD} multiplied by $\alpha$ and therefore holds. The second equation is also satisfied, since if $h_{xy}^\pi>0$ then also $\alpha h_{xy}^\pi>0$.
\item The complementary slackness condition is also satisfied by the new variables. Finally the condition $\alpha^{\gamma-1} \mu_{xy}^\pi\geq0$ is also satisfied since $\alpha>0$ and from the hypotheses $\mu_{xy}^\pi\geq0$.
\end{itemize}

In short, the homogeneity of the nonlinearity is the key property that allows just rescale the path flows to obtain a solution to the problem with rescaled demand matrix.
\end{proof}

\end{document}